\documentclass[journal,12pt,onecolumn,compsoc]{myIEEEtran}

\usepackage[T1]{fontenc}
\usepackage[utf8]{inputenc}
\usepackage{lmodern}

\usepackage[pdftex]{graphicx}
\usepackage{color}
\usepackage{longtable}

\usepackage[cmex10]{amsmath}
\usepackage{amssymb,amsthm}

\usepackage{algorithmic}
\usepackage{array}
\usepackage[caption=false,font=footnotesize]{subfig}
\usepackage{url}
\usepackage{enumitem}
\usepackage{multicol}
\usepackage{caption}

\newcommand{\ds}{\displaystyle}
\newcommand{\ts}{\textstyle}

\newcommand{\Rr}{{\mathbb R}}
\newcommand{\Cc}{{\mathbb C}}

\newcommand{\meanm}{\bar{\mathbf{m}}_{f}}
\newcommand{\meanc}{\bar{\mathbf{c}}_{g}}
\newcommand{\means}{\bar{\mathbf{s}}_{f,g}}
\newcommand{\meanr}{\bar{\mathbf{r}}_{f,g}^{(\theta)}}
\newcommand{\Mop}{\mathbf{M}_{f}}
\newcommand{\Cop}{\mathbf{C}_{g}}
\newcommand{\Sop}{\mathbf{S}_{f,g}}
\newcommand{\Rop}{\mathbf{R}_{f,g}^{(\theta)}}

\newcommand{\vvar}{\mathrm{var}}

\newcommand{\froot}{f^{1/2}}
\newcommand{\groot}{g^{1/2}}

\newcommand{\dist}{\mathrm{d}}
\newcommand{\distw}{\mathrm{d}_w}
\newcommand{\distwmax}{\mathrm{d}_w^{(\infty)}}

\theoremstyle{plain}
\newtheorem{theorem}{Theorem}[section]
\newtheorem{proposition}[theorem]{Proposition}

\newtheorem{lemma}[theorem]{Lemma}
\newtheorem{result}[theorem]{Property}

\theoremstyle{remark}
\newtheorem{remark}[theorem]{Remark}

\newtheorem{cexample}[theorem]{Counterexample}

\theoremstyle{definition}

\usepackage{tikz}
\usetikzlibrary{calc,fadings,decorations.pathreplacing}
\usetikzlibrary{matrix,arrows}

\def\minus{%
  \setbox0=\hbox{-}%
  \vcenter{%
    \hrule width\wd0 height \the\fontdimen8\textfont3%
  }%
}

\usepackage[ruled,vlined]{algorithm2e}

\begin{document}

\title{Shapes of Uncertainty in Spectral Graph Theory}

\author{Wolfgang Erb
\thanks{Universit{\`a} degli Studi di Padova, Dipartimento di Matematica ''Tullio Levi-Civita'', Padova, Italy, wolfgang.erb@lissajous.it.}
}

\markboth{Shapes of Uncertainty in Spectral Graph Theory}%
{Shapes of Uncertainty in Spectral Graph Theory}

\maketitle

\begin{abstract}
We present a flexible framework for uncertainty principles in spectral graph theory. In this framework, general filter functions modeling the spatial and spectral localization of a graph signal can be incorporated. It merges several existing uncertainty relations on graphs, among others the Landau-Pollak principle describing the joint admissibility region of two projection operators, and uncertainty relations based on spectral and spatial spreads. Using theoretical and computational aspects of the numerical range of matrices, we are able to characterize and illustrate the shapes of the uncertainty curves and to study the space-frequency localization of signals inside the admissibility regions. 
\end{abstract}


\IEEEpeerreviewmaketitle

\section{Introduction}

\IEEEPARstart{U}{ncertainty} principles are important cornerstones in signal analysis. They describe inherent limitations of a signal to be localized simultaneously in a complementary pair of domains, usually referred to as space and frequency (or spectral) domains. Uncertainty relations can be formulated in a multitude of ways. Exemplarily, the first uncertainty principle discovered by Heisenberg \cite{Heisenberg1927} can be written in terms of a commutator relation of a position with a momentum operator. In other contexts, uncertainty relations are described in terms of inequalities, by the space-frequency support or the smoothness of functions, or in form of boundary curves for an uncertainty region. A survey on different description possibilities can be found in \cite{FollandSitaram1997}, a wider theory is given in \cite{HavinJoericke}. 

In signal processing on graphs, uncertainty relations play a crucial role as well. They are used to describe the limitations of space-frequency localization \cite{AgaskarLu2013,StankovicDakovicSejdic2019}, to study sampling properties on graphs \cite{perraudin2018,TBL2016}, or to analyse space-frequency atoms and wavelet decompositions \cite{AgaskarLu2013,shuman2012,shuman2016}. The discrete harmonic structure in spectral graph theory (as introduced in \cite{Chung}) allows to transfer many uncertainty relations developed in classical settings directly onto a graph structure. This has led to a quite fragmented zoo of available uncertainty relations. As the discrete geometry of a graph can vary from a very homogeneous, symmetric geometry to a very inhomogeneous structure, not every uncertainty principle is equally useful for every graph. To give an example, it is shown in \cite{BenedettoKoprowski2015,StankovicDakovicSejdic2019} that particular graphs, as for instance complete graphs, exhibit a harmonic structure in which the support theorem of Elad and Bruckstein \cite{EladBruckstein2002} provides only a very week uncertainty relation. 
For graphs it is therefore important to have a flexible framework of uncertainty principles at hand that can be adapted to the graph structure or to particular applications in graph signal processing. 

In spectral graph theory, the complementary pair of domains in which localization is measured is given by a space domain consisting of a discrete set of graph nodes and a spectral domain provided by the eigendecomposition of a graph Laplacian \cite{Chung}. The description of uncertainty principles on graphs relies on this graph-dependent spatial and spectral structure. The second key ingredient for the formulation of an uncertainty principle is a proper concept for the measurement of space and frequency localization. Suitable localization measures on graphs should be consistent with the given harmonic structure, but also be adjustable to prerequisites determined by applications.

The goal of this article is to offer a new more general perspective on uncertainty principles in spectral graph theory, able to incorporate a large number of different localization measures. These measures will be defined in terms of localization operators built upon a space and a frequency filter. For any given pair of filters we want to visualize and characterize the shapes of the corresponding uncertainty regions. For this, we combine  existing results on operator-based uncertainty principles in signal processing with computational methods developed for the numerical range of matrices. In this way, we get a powerful unified framework for the analysis, the computation and the illustration of uncertainty in spectral graph theory.

\noindent \textbf{Main contributions.} Our uncertainty framework on graphs is a synthesis and extension of several established results. It is built on the following ideas, compactly illustrated in Table 1.
\begin{itemize}
\item The theoretical starting point of this framework is the space-frequency analysis studied by Landau, Pollak and Slepian \cite{LandauPollak1961,LandauPollak1962,LandauWidom1980,Slepian1978,Slepian1983,SlepianPollak1961} for signals on the real line and its relative on graphs \cite{TBL2016,VanDeVille2016}. In Sections 
\ref{sec:spacefrequency} and \ref{sec:spacefrequencyLPS}, we will extend this theory from projection operators to general symmetric, positive-semidefinite operators. In this way, we are able to merge the Landau-Pollak uncertainty for projection operators with uncertainty relations based on spectral and spatial spreads on graphs as formulated in \cite{AgaskarLu2013,pasdeloup2015,pasdeloup2019uncertainty}. Our main new theoretical result in this part is the uncertainty estimate given in Theorem \ref{Theorem-uncertainty}. 
\item The second key technology implemented in this framework consists of theoretical and computational aspects of the numerical range of matrices \cite{Brickman,Hausdorff1919,HornJohnson1991,Johnson1978,Lenard1972,Toeplitz1918}. In Section \ref{sec:numericalrange}, we show that the uncertainty regions related to the space and frequency localization measures can be formulated and calculated with help of a convex numerical range. In this way, we obtain the efficient Algorithm 1 for the computation and visualization of the uncertainty curves and, in addition, the theoretical bounds in Theorems \ref{cor-UP-rotatedoperator} and \ref{Theorem-uncertaintyconvex} for the uncertainty regions. This simplifies the methods derived in \cite{AgaskarLu2013} for the calculation of the convex uncertainty curve.
\item In a Landau-Pollak-Slepian type space-frequency analysis, it is possible to formulate error estimates for the approximation of space-frequency localized signals \cite{erb2013,ErbMathias2015,LandauPollak1962}. In our general framework on graphs, these estimates will be derived in Section \ref{sec:errorestimates}.
\end{itemize}

Compared to existing uncertainty relations on graphs that are based on predefined single space-frequency filters, the filter pairs in our general framework are flexible. This gives the interesting opportunity to design filter functions for a graph-adapted space-frequency analysis. In particular, the generality of our framework has the following advantages:
\begin{itemize}
\item The Landau-Pollak-Slepian space-frequency analysis is based on a set-oriented localization while the spectral spreads defined in \cite{AgaskarLu2013} favor signals localized at the lower end of the graph spectrum. In our framework arbitrary distance functions on the graph or its spectrum can be implemented to measure different types of space-frequency localization. This allows us to design localization measures that contain additional spatial, directional or spectral information of the graph. We will provide some examples in Section \ref{sec:shapeexamples}.
\item By using projection filters in the Landau-Pollak-Slepian setting the spectrum of the space-frequency operator clusters at the values $0$ and $1$ \cite{LandauWidom1980,Slepian1983}. This leads to numerical instability when calculating the corresponding eigendecomposition directly. In Section \ref{sec:shapeexamples}, we will observe a similar clustering for graphs. By the usage of alternative filters, this clustering can be avoided. 
\item In graph convolutional neural networks, spatial and spectral filters are extracted in an optimization process \cite{Defferrard2016,Tran2018}. The generality of our framework allows to express uncertainty relations also for such learned filter functions. 
\end{itemize} 

\begin{table} \centering
\begin{tikzpicture}[
node distance = 2cm,
base/.style = {draw, minimum width = 3cm, minimum height = 1cm, align = center},
start/.style = {base, rectangle,rounded corners, fill = red!30},
input/.style = {base, rectangle,rounded corners, fill = blue!30},
]

\node (node0) [start] {Unified framework \\ of uncertainty principles \\ in spectral graph theory};
\node (node1) [input, above of=node0,yshift=1cm] 
{Landau-Pollak-Slepian theory \\ 
In signal processing: \cite{LandauPollak1961,Slepian1983,SlepianPollak1961} \\
On graphs: \cite{TBL2016}};
\node (node2) [input, left of=node0,xshift=-4.5cm] {Numerical range \\ of matrices \\ \cite{Brickman,Johnson1978,Lenard1972,Toeplitz1918}};
\node (node3) [input, right of=node0,xshift=5cm] {Uncertainty principle\\ based on graph spreads\\ and spectral spreads \cite{AgaskarLu2013}};

\draw [thick,arrows=-stealth] (node1) --node[anchor=east] {\scriptsize $\begin{array}{l}  \text{provides theory} \\ \text{and estimates} \\ \text{for boundary} \end{array}$} (node0);
\draw [thick,arrows=-stealth] (node1) --node[anchor=west] {\scriptsize $\begin{array}{l}  \text{provides concept for} \\ \text{space-frequency analysis} \end{array}$} (node0);
\draw [thick,arrows=-stealth] (node2) --node[anchor=north] { \scriptsize $\begin{array}{l}  \text{computational} \\ \text{methods} \end{array}$} (node0);
\draw [thick,arrows=-stealth] (node2) --node[anchor=south] { \scriptsize $\begin{array}{l} \text{results on} \\  \text{approximation} \\  \text{and boundary} \end{array}$} (node0);
\draw [thick,arrows=-stealth] (node3) --node[anchor=south] { \scriptsize $\begin{array}{l} \text{particular} \\  \text{filter functions} \end{array}$} (node0);
\draw [thick,arrows=-stealth] (node3) --node[anchor=north] { \scriptsize $\begin{array}{l} \text{description of} \\  \text{uncertainty} \\  \text{curve} \end{array}$} (node0);
\end{tikzpicture}
\caption{Main conceptual influences for the uncertainty framework studied in this work.}
\end{table}

\section{Background}
\subsection{Spectral graph theory}
The goal of this section is to give a broad overview on spectral graph theory and 
the notion of harmonic analysis on a graph $G$, essential for the formulation of space-frequency 
decompositions or uncertainty principles. A profound introduction to spectral graph theory can be found in \cite{Chung}. For an introduction to the Fourier transform and space-frequency concepts on graphs, we refer to \cite{shuman2016}. 

We describe the graph $G$ as a triplet $G=(V,E,\mathbf{A})$, where $V=\{v_1, \ldots, v_{n}\}$ denotes the set of vertices (or nodes) of the graph, $E \subseteq V \times V$ is the set of (directed or undirected) edges connecting the nodes and $\mathbf{A} \in \mathbb{R}^{n \times n}$ is a weighted, symmetric
and non-negative adjacency matrix containing the connection weights of the edges. The entire harmonic structure of the graph $G$ is encoded and described by this adjacency matrix $\mathbf{A}$. Note that, although $G$ can also be a directed graph, the symmetric matrix $\mathbf{A}$ gives an undirected harmonic structure on $G$. 

We aim at studying signals $x$ on the graph $G$, i.e. functions $x: V \rightarrow \mathbb{R}$ that associate a real value to each node of $V$. Since the number of nodes in $G$ is fixed (i.e. $n$) and the set $V$ is ordered, we can naturally represent the signal $x$ as a vector 
$x = (x_1, \ldots, x_n)\in \mathbb{R}^n$. Depending on the context, we will switch between these two 
representations. 

To define the Fourier transform on $G$, we consider
the (normalized) graph Laplacian $\mathbf{L}$ associated to the adjacency matrix $\mathbf{A}$:
\begin{equation*}
    \mathbf{L} := \mathbf{I_n} - \mathbf{T}^{-\frac{1}{2}}\mathbf{A}\mathbf{T}^{-\frac{1}{2}}.
\end{equation*}
Here, $\mathbf{I_n}$ is the $n \times n$ identity matrix, and $\mathbf{T}$ is the degree matrix with entries given as
\begin{equation*}
    (\mathbf{T})_{ij} := \left.
  \begin{cases}
    \sum_{k=0}^n (\mathbf{A})_{ik}, & \text{if } i=j \\
    0, & \text{otherwise}
  \end{cases}.
  \right.
\end{equation*}
Since the adjacency matrix $\mathbf{A}$ is symmetric, also $\mathbf{L}$ is a symmetric operator and we can compute its orthonormal eigendecomposition as
\begin{equation*}
\mathbf{L}=\mathbf{U}\mathbf{M}_{\lambda} \mathbf{U^\intercal},
\end{equation*}
where 
$\mathbf{M}_{\lambda} = \mathrm{diag}(\lambda) = \text{diag}(\lambda_1,\ldots,\lambda_{n})$ 
is the diagonal matrix with the increasingly ordered eigenvalues $\lambda_i$, $i \in \{1, \ldots, n\}$, of $\mathbf{L}$ as diagonal entries, i.e.,
$$(\mathbf{M}_{\lambda})_{ij} := 
\big(\mathrm{diag}(\lambda)\big)_{ij} =\begin{cases}
\lambda_i & \text{if }i=j\\
0 & \text{otherwise}
\end{cases}.
$$
The columns $ u_1, \ldots, u_{n}$ of the orthonormal matrix $\mathbf{U}$ are normalized eigenvectors of
$\mathbf{L}$ with respect to the eigenvalues $\lambda_1, \ldots, \lambda_n$. The ordered set $\hat{G} = \{u_1, \ldots, u_{n}\}$ of eigenvectors is an orthonormal basis for the space of signals on the graph $G$. We call $\hat{G}$ the spectrum of the graph $G$.

\subsection{Fourier transform on graphs} 
In classical Fourier analysis, as for instance the Euclidean space or the torus, the Fourier transform can be defined in terms of the eigenvalues and eigenfunctions of the Laplace operator.
In analogy, we consider the elements of $\hat{G}$, i.e. the eigenvectors $\{u_1, \ldots, u_{n}\}$, as the Fourier basis on the graph $G$. In particular, going back to our 
spatial signal $x$, we can define the graph Fourier transform of $x$ as
\begin{equation*}
\hat{x} := \mathbf{U^\intercal}x,
\end{equation*}
and its inverse graph Fourier transform as
\begin{equation*}
x := \mathbf{U}\hat{x}. 
\end{equation*}
The entries $\hat{x}_i = u_i^\intercal x$ of $\hat{x}$ are the frequency components or coefficients
of the signal $x$ with respect to the basis function $u_i$. For this reason, $\hat{x} : \hat{G} \to \Rr$ can be seen as a distribution on the spectral domain $\hat{G}$ of the graph $G$. 
To keep the notation simple, we will however usually represent spectral distributions $\hat{x}$ as 
vectors $(\hat{x}_1, \ldots, \hat{x}_n)$ in $\Rr^n$. Regarding the eigenvalues of the normalized
graph Laplacian $\mathbf{L}$ it is well-known that (see \cite[Lemma 1.7]{Chung})
\[0 = \lambda_1 \leq \lambda_2 \leq \cdots \leq \lambda_n \leq 2.\]

\subsection{Spatial and spectral filtering on graphs}

By using the graph Fourier transform to switch between spatial and spectral domain, we can now define (pointwise) multiplication and convolution between two signals $x$ and $y$. As there is no immediate description of translation on $G$, it is easier to define the convolution in the spectral domain, using an analogy to classical Fourier analysis in which the convolution of two signals is calculated as the pointwise product of their Fourier transforms. We define
\begin{equation}
    x * y := \mathbf{U} \left ( \hat{x}   \odot  \hat{y} \right ) = \mathbf{U} \left ( \left ( \mathbf{U^\intercal}x \right )  \odot \left ( \mathbf{U^\intercal} y \right )\right ) \label{eq:spectralfilter1},
\end{equation}
where $\hat{x} \odot \hat{y} := (\hat{x}_1 \hat{y}_1, \ldots, \hat{x}_{n} \hat{y}_{n})$
denotes the pointwise Hadamard product of the two vectors $\hat{x}$ and $\hat{y}$. 
The Hadamard product $x \odot y$ between two signals $x$ and $y$ can be formulated in matrix-vector notation as 
$x \odot y = \mathbf{M}_x y$ by applying the diagonal matrix 
$\mathbf{M}_x = \mathrm{diag}(x)$ to the vector $y$. In the same way, we get also for functions $\hat{x}$ and $\hat{y}$ on $\hat{G}$ the notion $\hat{x} \odot \hat{y} = \mathbf{M}_{\hat{x}} \hat{y}$. 
In this way, according to \eqref{eq:spectralfilter1}, we obtain for the convolution the identities
\begin{equation*}
    x * y = \mathbf{U}\mathbf{M}_{\hat{x}}\mathbf{U^\intercal}y = \mathbf{U}\mathbf{M}_{\hat{y}}\mathbf{U^\intercal} x.
\end{equation*}

\section{Space and frequency localization on graphs} \label{sec:spacefrequency}
\subsection{General setting}
We are going to study a space-frequency analysis for signals $x$ on the graph $G$ based on two
nonnegative normalized filter functions $f,g \in \Rr^n$ with the properties
\begin{equation} \label{eq:propspacefreq}
0 \leq f,g \leq 1 \quad \text{and} \quad \|f\|_\infty = \|g\|_\infty = 1.
\end{equation} 
Given the two filters $f$ and $g$ we introduce the following two operators
\begin{align*}
 \mathbf{M}_f x  &:= f \odot x, \\
 \mathbf{C}_g x  &:= g * x = \mathbf{U}\mathbf{M}_{\hat{g}}\mathbf{U^\intercal} x.
\end{align*}
The point-wise multiplication $\mathbf{M}_f$ with the filter $f$ will be referred to as space localization operator, the convolution $\mathbf{C}_g$ as frequency localization operator. From
the properties of $f$ and $g$ in \eqref{eq:propspacefreq} it follows immediately that both operators
$\mathbf{M}_f$ and $\mathbf{C}_g$ are symmetric and positive-semidefinite and the spectral norm of both operators is exactly $1$. For the operators $\mathbf{M}_f$ and $\mathbf{C}_g$ we define the expectation
values
\begin{align*}
\meanm(x) := \frac{\langle \mathbf{M}_f x, x \rangle }{\|x\|^2},\qquad 
\meanc(x) := \frac{\langle \mathbf{C}_g x, x \rangle}{\|x\|^2}.
\end{align*}

We say that a signal $x$ on $G$ is space-localized with respect to the 
window function $f$ if $\meanm(x)$ is close to one. In the same way, we say that 
$x$ on $G$ is frequency-localized with respect to $g$ 
if $\meanc(x)$ approaches one. Based on the mean values $\meanm(x)$ and $\meanc(x)$, we define the set of admissible values related to the operators $\mathbf{M}_f$ and $\mathbf{C}_g$ as 
\begin{equation} \label{eq:numericalrange} \mathcal{W}(\mathbf{M}_f, \mathbf{C}_g) := \left\{(\meanm(x),\meanc(x)) \; | \; \|x\| = 1 \right\} \subset [0,1]^2. 
\end{equation}
Due to a relation profoundly described in Section \ref{sec:numericalrange}, we call $\mathcal{W}(\mathbf{M}_f, \mathbf{C}_g)$ also the numerical range of the pair $(\mathbf{M}_f, \mathbf{C}_g)$ of operators. All the uncertainty principles studied in this work are linked to the boundaries of the set $\mathcal{W}(\mathbf{M}_f, \mathbf{C}_g)$.

\subsection{Space-frequency operators on graphs}
To investigate the joint localization of a signal $x$ with
respect to the spatial filter $f$ and the frequency filter $g$, as well as for the description of the set $\mathcal{W}(\mathbf{M}_f, \mathbf{C}_g)$, we consider the
following two space-frequency operators on the graph $G$: 
\begin{equation*}
\Rop := \cos(\theta) \, \mathbf{M}_f + \sin (\theta) \, \mathbf{C}_{g} \quad \text{and} \quad
\Sop := \mathbf{C}_{\groot} \mathbf{M}_f \mathbf{C}_{\groot}.
\end{equation*}
The linear combination $\Rop$ of the two symmetric matrices $\Mop$ and $\Cop$ is symmetric for any angle $0 \leq \theta < 2 \pi$ and positive semidefinite if $0 \leq \theta \leq \frac{\pi}{2}$. The operator norm of $\Rop$ is bounded by $|\cos \theta| + |\sin \theta| \leq \sqrt{2}$. 
The composition $\mathbf{S}_{f,g} \in \Rr^{n \times n}$ is a positive semidefinite, symmetric matrix with spectral norm bounded by $1$.

To study the space-frequency operators $\Rop$ and $\Sop$ we focus on their eigendecompositions
\[\Rop = \mathbf{\Phi}\mathbf{M}_{\rho}\mathbf{\Phi^\intercal},\quad \text{and} \quad \mathbf{S}_{f,g} = \mathbf{\Psi} \mathbf{M}_{\sigma}\mathbf{\Psi^\intercal}. \]
The decreasingly ordered eigenvalues $\rho = \rho^{(\theta)} = (\rho_1^{(\theta)}, \ldots \rho_n^{(\theta)})$ of the matrix $\Rop$ are real and contained in $[-\sqrt{2}, \sqrt{2}]$. The decreasingly ordered eigenvalues $\sigma = (\sigma_1, \ldots \sigma_n)$ of $\Sop$, are non-negative and smaller than $1$. The columns $\{\phi_1^{(\theta)}, \ldots, \phi_{n}^{(\theta)}\}$ of $\mathbf{\Phi}$ and $\{\psi_1, \ldots, \psi_{n}\}$ of $\mathbf{\Psi}$ form a complete set of orthonormal eigenvectors of the operators $\Rop$ and $\Sop$, respectively.  We say that a signal $x$ is space-frequency localized with respect to the filters $f$ and $g$ if the expectation values
$$ \meanr(x) := \frac{\langle \Rop x, x \rangle}{\|x\|^2}
= \cos(\theta) \, \meanm(x) + \sin(\theta) \, \meanc(x) \quad \text{and} \quad \means(x) := \frac{\langle \mathbf{S}_{f,g} x, x \rangle}{\|x\|^2}$$
get close to one. The largest eigenvalues $\rho_1^{(\theta)}$ and $\sigma_1$
and the corresponding eigenvectors $\phi_1^{(\theta)}$ and $\psi_1$,
will be of major importance of us. For the largest eigenvalue $\sigma_1$ of the space-frequency operator $\mathbf{S}_{f,g}$ we get additionally the following characterizations.

\begin{result} \label{res:sigma}
The largest eigenvalue $\sigma_1$ of the space-frequency operator $\mathbf{S}_{f,g}$
corresponds to the following spectral operator norms:
\[\sigma_1 = \|S_{f,g}\| = \|\mathbf{M}_{\froot} \mathbf{C}_{\groot}\|
^2 = \| \mathbf{C}_{\groot} \mathbf{M}_{\froot}\|^2 = \| \mathbf{M}_{\froot} \mathbf{C}_{g} \mathbf{M}_{\froot}\|.\]
\end{result}

\subsection{Examples of space-frequency filters on graphs} \label{sec:examples}

\begin{enumerate}[label = (\arabic*)]
\item (Landau-Pollak-Slepian filters or projection-projection filters) Let $\chi_{\mathcal{A}}$ denote the indicator function of a set $\mathcal{A}$, i.e.
\[ \chi_{\mathcal{A}}(v) := \left\{ \begin{array}{ll} 1 & \text{if $v \in \mathcal{A}$,} \\
0 & \text{if $v \notin \mathcal{A}$.} \end{array} \right. \]
For a subset $\mathcal{A}$ of the node set $V$ and a subset $\mathcal{B}$ of the spectrum $\hat{G}$, we define the filter functions $f$ and $g$ as
\begin{equation} \label{eq:slepianfilter}
f = \chi_{\mathcal{A}} \quad \hat{g} = \chi_{\mathcal{B}}.
\end{equation}
Both, the space-localization $\mathbf{M}_f$ and the frequency-localization operator $\mathbf{C}_g$ defined in terms of $f$ and $\hat{g} = \mathbf{U}^{\intercal} g$ in \eqref{eq:slepianfilter} are projection operators satisfying
$\mathbf{M}_f^2 = \mathbf{M}_f$ and $\mathbf{C}_g^2 = \mathbf{C}_g$. The space-frequency operator
$\mathbf{S}_{f,g}$ is in this case equivalently given as $\mathbf{S}_{f,g} = \mathbf{C}_{g} \mathbf{M}_{f} \mathbf{C}_g$. For signals on the real line, the space-frequency analysis related to these 
projection operators, including the study of uncertainty principle and the distribution of the eigenvalues $\sigma_i$ was studied intensively by Landau, Pollak and Slepian in a series of papers in the sixties of the last century, cf. \cite{LandauPollak1961,LandauPollak1962,Slepian1978,SlepianPollak1961}. A variant of this theory on the unit sphere is given in \cite{PlattnerSimons2014,SimonsDahlenWieczorek2006}. A general theory based on two projection operators in an abstract Hilbert spaces, can be found in \cite[Chapter 3 \S 1]{HavinJoericke}. The translation to the graph setting was conducted in \cite{TBL2016}.
 
\item (Distance-projection filters) As a second spatial filter, we want to generate a window function that limits a signal $x$ if the distance $\dist(v,w)$ to a point $w$ on the graph gets large. In general, $\dist(v,w)$ can be any distance metric on the nodes $V$ of the graph. In this article, we will use the geodesic distance on the graph as a metric $\dist$, i.e. $\dist(v,w)$ is the length of the shortest path connecting the nodes $v$ and $w$. We further set
$$\distw(v) := \dist(v,w), \quad \distwmax := \max_{v \in V} \dist(v,w).$$
Then, as spatial filter $f$ and frequency filter $g$, we define
\begin{equation} \label{eq:distprojfilter}
f(v) = 1 - \distw(v)/\distwmax,\; v \in V, \quad \text{and} \quad \hat{g} = \chi_{\mathcal{B}},
\end{equation}
i.e., the spatial filter $f$ incorporates the distance $\distw$ to a reference node $w$ and the frequency filter $g$, as before, describes the projection on a spectral subset $\mathcal{B} \subset \hat{G}$. For  $f$ we have 
\[ \mathbf{M}_f x = \mathbf{I}_n - \ts \frac{1}{\distwmax} \mathbf{M}_{\distw} x, \quad 
\meanm(x) = 1 - \ts \frac{ \ds x^{\intercal} \mathbf{M}_{\distw}x}{\distwmax \|x\|^2}.\]
Similar distance-projection filters were used in a 
continuous setup for orthogonal expansions on the interval $[-1,1]$ \cite{erb2012,erb2013,Klaja2016} and on the unit sphere \cite{ErbMathias2015}.
\item (Modified distance-projection filters)
While the projection filter $\hat{g} = \chi_{\mathcal{B}}$ generates bandlimited signals on the graph $G$ (with support in the frequency "band" $\mathcal{B} \subset \hat{G}$), in applications it is often relevant to additionally soften the higher frequencies. 
This can be achieved by multiplying the projection filter $\chi_{\mathcal{B}}$ in the spectral domain with a filter function $\hat{g}^{(\beta)}$ in which the coefficients $0 \leq \hat{g}_k^{(\beta)} \leq 1$ decay for increasing frequency $k$. Also for the distance filter $f$ given in \eqref{eq:distprojfilter} slight
modifications can be useful in order to alter the influence of the distance. A simple possibility here is
to add an additional power $\alpha>0$ to the distance function. 
In this way, we get as modified filters 
\begin{equation} \label{eq:distprojsmoothfilter}
f(v) = 1 - \left(\frac{\distw(v)}{\distwmax}\right)^{\alpha}, \quad \hat{g} = \chi_{\mathcal{B}} \odot \hat{g}^{(\beta)}.
\end{equation}
The effects of such modifications to the shape of the uncertainty principles will be investigated in  Section \ref{sec:shapeexamples}. An example of such a modified filter function in the frequency domain is given in \cite{VanDeVille2016}.
Here, the author proposes to use the eigenvalues $\lambda_i$ of the graph Laplacian to 
define the components $\hat{g}^{(\beta)}_i$ of the additional filter.
\item (Distance-Laplace filter)
Another spectral filter $\hat{g} = (\hat{g}_1, \ldots, \hat{g}_n)$ on $\hat{G}$ can be defined as  
\begin{equation} \label{eq:laplacefilter}  \hat{g}_k = 1 - \lambda_k/2, \end{equation}
where $\lambda_k$ denotes the $k$-th. smallest eigenvalue of the graph Laplacian $\mathbf{L}$. 
In this case, we get
\[ \mathbf{C}_g x = \mathbf{U} (\mathbf{I}_n - \ts \frac{1}{2} \mathbf{M}_{\lambda}) \mathbf{U}^{\intercal} x
= (\mathbf{I}_n - \ts \frac{1}{2}\mathbf{L}) x. \]
Dividing by the factor $2$ in \eqref{eq:laplacefilter} guarantees that the spectrum of $\mathbf{C}_g$ is contained
in $[0,1]$ (as the spectrum of the normalized Laplacian $\mathbf{L}$ is contained in $[0,2]$).  
Using the modified distance filter from \eqref{eq:distprojsmoothfilter} with $\alpha = 2$ as a spatial filter, we get 
\[ \meanm(x) = 1 - \ts \frac{x^{\intercal} \mathbf{M}_{\distw^2}x}{(\distwmax)^2 \|x\|^2} , \qquad \meanc(x) = 1 - \frac{x^{\intercal} \mathbf{L} x}{2 \|x\|^2} .\]
In \cite{AgaskarLu2013}, the measure $\bar{\mathbf{c}}_{\mathbf{U} \lambda}(x) = x^{\intercal} \mathbf{L} x /\|x\|^2$ is called the spectral spread of $x$ while the measure 
$\bar{\mathbf{m}}_{\distw^2}(x) = x^{\intercal} \mathbf{M}_{\distw^2}x / \|x\|^2$ is denoted as the graph spread of $x$ with respect to the node $w$. In \cite{AgaskarLu2013} an uncertainty principle on graphs is formulated in terms of these two spreads. We will show that this uncertainty principle fits as a special case in our more general framework. 
\item (Laplace-Laplace filter) Instead of defining the operators $\mathbf{M}_f$ and $\mathbf{C}_g$ for signals on the graph domain, it is also possible to define them in terms of distributions $\hat{x}$ in the spectral domain $\hat{G}$. In this sense, we can introduce the operators $\mathbf{M}_{\hat{f}}$ and $\mathbf{C}_{\hat{g}}$ as
\[\mathbf{M}_{\hat{f}} \hat{x} = \hat{f} \odot \hat{x}, \quad \mathbf{C}_{\hat{g}} \hat{x} = 
  \hat{g} \ast \hat{x} = \mathbf{U} \mathbf{M}_{\hat{\hat{g}}} \mathbf{U}^{\intercal} \hat{x}.\]
An example of such a filtering is given in \cite{BenedettoKoprowski2015}. Here, $f$ and $g$ are given as
\[ \hat{f} = 1 - \ts \frac12 \lambda, \qquad \hat{g} = 1 - \ts \frac12 \lambda.\]
In \cite{BenedettoKoprowski2015}, the filters are actually formulated in terms of $\hat{f} = \hat{g} = \lambda$. The reformulation above guarantees that the entries $\hat{f}$ and $\hat{g}$ are between $0$ and $1$ and that the relevant part of the uncertainty curve is located at $(1,1)$ instead of $(0,0)$. 
\end{enumerate}

\subsection{Comparison to space-frequency analysis based on windowed Fourier transform}

An interesting source for uncertainty principles and space-frequency analysis on graphs is based on the windowed Fourier transform \cite{perraudin2018,shuman2012,shuman2016}. For this, we want to give a brief comparison between the space-frequency concepts studied in this article and those related to the windowed Fourier transform. For a window function $h: G \to \Rr$, the windowed Fourier transform 
$\mathbf{F}_{h} x $ of a signal $x$ is defined in the domain $G \times \hat{G}$ as 
\begin{equation} \label{eq:windowedFouriertransform}
\mathbf{F}_{h} x (v_i,u_k) := x^{\intercal} (\mathbf{M}_{u_k} \mathbf{C}_{e_i} h),
\end{equation}
where $e_i = \chi_{v_i}$, $i \in \{1, \ldots, n\}$ simply denote the standard basis vectors for the space of signals on $G$. In this definition, $\mathbf{C}_{e_i} h$ can be interpreted as a generalized shift of the window $h$ on $G$ whereas $\mathbf{M}_{u_k}$ describes a generalized modulation in terms of the Fourier
basis $u_k$. The space-frequency analysis related to the windowed Fourier transform uses 
the coefficients $\mathbf{F}_{h} x (v_i,u_k)$ to analyse the signal $x$. Further, the system 
$\{\mathbf{M}_{u_k} \mathbf{C}_{e_i} h\}_{i,k}$ provides a frame for the space of signals on $G$ if
$\hat{h}_1 \neq 0$. Compared to the space-frequency analysis studied in this paper, there are some conceptual differences:
\begin{enumerate}
\item The windowed Fourier transform is based on the choice of a single window function $h$, compared to the two filters $f$ and $g$ for the operators $\Sop$ and $\Rop$.
\item The space-frequency analysis related to the windowed Fourier transform is based on the frame system $\{\mathbf{M}_{u_k} \mathbf{C}_{e_i} h\}_{i,k}$ (cf. \cite{shuman2016}) compared to the orthogonal basis of eigenfunctions $\{\psi_k\}$ and $\{ \phi_k^{(\theta)}\}$ of the operators $\Sop$ and $\Rop$.
\item Uncertainty relations are formulated in terms of the frame coefficients in \eqref{eq:windowedFouriertransform}, see \cite{perraudin2018}.
\end{enumerate}
There are several possibilities to vary the definition in \eqref{eq:windowedFouriertransform}, leading to a similar space-frequency analysis. A simple example here is to exchange the order of $\mathbf{M}_{u_k}$ and $\mathbf{C}_{e_i}$ in \eqref{eq:windowedFouriertransform}. More detailed discussions about the windowed Fourier transform and frame decompositions can be found in \cite{perraudin2018,shuman2016}.

\section{Uncertainty principles related to the operator $\mathbf{S}_{f,g}$} 
\label{sec:spacefrequencyLPS}
We start with a first uncertainty relation for general filter functions $f$ and $g$
with the normalization \eqref{eq:propspacefreq} that rely on the maximal eigenvalue $\sigma_1$ of the space-frequency operator $\Sop$. This type of uncertainty principle was first studied by Landau and Pollak \cite{LandauPollak1961} for projection operators acting on functions on the real line. The corresponding relation for projection operators on graphs was formulated in \cite{TBL2016}. In our setting, this corresponds to filter functions $f$ and $g$ defined in terms of an indicator function, i.e. the setting of Example (1) in Section \ref{sec:examples}. Goal of this section is to prove this uncertainty principle now for general localization operators $\mathbf{M}_f$ and $\mathbf{C}_g$.

If $\|\mathbf{S}_{f,g}\| = \sigma_{1} < 1$, we can specify this uncertainty relation on $G$ by providing an explicit bound for the admissibility region $\mathcal{W}(\mathbf{M}_f, \mathbf{C}_g)$. This bound is  based on the curve
\[
\gamma_{f,g}: [\sigma_{1},1] \to \Rr:\quad \gamma_{f,g}(t) := \big((t \, \sigma_{1})^{\frac{1}{2}} + 
( (1-t)(1-\sigma_{1}) )^{\frac{1}{2}}\big)^2.
\]

\begin{lemma} \label{lemma-uncertaintyangle}
Assume that $\sigma_{1} < 1$. If $ \meanm(x) \meanc(x) \geq \sigma_{1}$, then the inequality
\begin{equation} \label{equation-uncerataintyangle1a}
\arccos \frac{\meanm(x)}{\|\mathbf{M}_f x\|} + \arccos \frac{\meanc(x)}{\|\mathbf{C}_g x\|} \geq \arccos \frac{\sqrt{\meanm(x)} \sqrt{\meanc(x)} }{\|\mathbf{M}_f x\| \|\mathbf{C}_g x\|} \sqrt{\sigma_{1}}
\end{equation}
holds true. This implies the upper bound 
\begin{equation} \label{equation-uncerataintyangle1b}
\meanc(x) \leq \gamma_{f,g}(\meanm(x)) = \big((\meanm(x) \sigma_{1})^{\frac{1}{2}} + ((1-\meanm(x))
(1-\sigma_{1}))^{\frac{1}{2}}\big)^2
\end{equation} 
for $\meanc(x)$ in the domain $\meanm(x) \meanc(x) \geq \sigma_{1}$. 
\end{lemma}

\begin{proof}
Let $x$ be a normalized signal on the graph $G$ with $\|x\| = 1$. Further, we consider the two
normalized vectors
$$y_1 = \frac{1}{\|\mathbf{M}_{f} x\|} \mathbf{M}_{f} x \quad \text{and} \quad 
y = \frac{1}{\| \mathbf{C}_g x\|} \mathbf{C}_g x.$$ 
The angular distance is a metric for vectors on the unit sphere. In particular, the sum of the angular distances between the vectors $y_1$ and $x$, and $y_2$ and $x$ is always larger than the angular distance between $y_1$ and $y_2$, i.e.
\begin{equation} \label{equation-proofuncertainty1}
\arccos \langle y_1, x \rangle + \arccos \langle y_2, x \rangle \geq \arccos \langle y_1,y_2 \rangle.
\end{equation}
For the term $\langle y_1,y_2 \rangle$, we can find an upper bound using the Cauchy-Schwarz-inequality:
\begin{align*}
\langle y_1,y_2 \rangle & \leq |\langle y_1,y_2 \rangle| = \frac{|\langle \mathbf{M}_f x, \mathbf{C}_g x \rangle|}{\|\mathbf{M}_f x\| \|\mathbf{C}_g x\|} = \frac{|\langle \mathbf{M}_{f^{1/2}} x, \mathbf{M}_{f^{1/2}} \mathbf{C}_g x \rangle|}{\|\mathbf{M}_f x\| \|\mathbf{C}_g x\|}\\ &\leq
\frac{\|\mathbf{M}_{f^{1/2}} x \| \| \mathbf{M}_{f^{1/2}} \mathbf{C}_g x \|}{\|\mathbf{M}_f x\| \|\mathbf{C}_g x\|} 
= \frac{\sqrt{\meanm(x)} \sqrt{ \langle \mathbf{C}_{g^{1/2}} \mathbf{M}_{f} \mathbf{C}_g x, \mathbf{C}_{g^{1/2}} x \rangle}}{\|\mathbf{M}_f x\| \|\mathbf{C}_g x\|} \\
& \leq \frac{\sqrt{\meanm(x)} \sqrt{\meanc(x)}\sqrt{ \langle \mathbf{C}_{g^{1/2}} \mathbf{M}_{f} \mathbf{C}_{g^{1/2}} x, x \rangle}}{\|\mathbf{M}_f x\| \|\mathbf{C}_g x\|}
\leq \frac{\sqrt{\meanm(x)} \sqrt{\meanc(x)} }{\|\mathbf{M}_f x\| \|\mathbf{C}_g x\|} \sqrt{\sigma_{1}}.
\end{align*}
As we assume that $ \meanm(x) \meanc(x) \geq \sigma_{1}$ the last expression is smaller than 1. 
Therefore in \eqref{equation-proofuncertainty1}, we get 
\begin{align*}
\langle y_1,y_2 \rangle \leq \frac{\sqrt{\meanm(x)} \sqrt{\meanc(x)} }{\|\mathbf{M}_f x\| \|\mathbf{C}_g x\|} \sqrt{\sigma_{1}} \leq 1, \quad \langle y_1,x \rangle = \frac{\meanm(x)}{\|\mathbf{M}_f x\|}, \quad
\langle y_2,x \rangle &= \frac{\meanc(x)}{\|\mathbf{C}_g x\|},
\end{align*}
and, thus, precisely the inequality \eqref{equation-uncerataintyangle1a}
To demonstrate the second inequality we make use of the following fact:
\begin{equation} \label{eq:fact}
\text{if $0 < a \leq b \leq 1$, then $\arccos bt - \arccos at$ is a decreasing function in $t \in \ts \left[-\frac{1}{b}, \frac{1}{b}\right]$.}
\end{equation}
Therefore, by setting $a = \frac{\sqrt{\meanc(x)} }{\|\mathbf{C}_g x\|} \sqrt{\sigma_{1}} \leq  \frac{\meanc(x) }{\|\mathbf{C}_g x\|}\sqrt{\meanm(x)} \leq \sqrt{\meanm(x)} = b\leq 1$, we can apply 
\eqref{eq:fact} to inequality \eqref{equation-uncerataintyangle1b} and obtain
\begin{equation*} 
\arccos \sqrt{\meanm(x)}+ \arccos \frac{\meanc(x)}{\|\mathbf{C}_g x\|} \geq \arccos \frac{ \sqrt{\meanc(x)} }{ \|\mathbf{C}_g x\|} \sqrt{\sigma_{1}}.
\end{equation*} 
Applying \eqref{eq:fact} a second time with $a = \sqrt{\sigma_{1}} \leq  \sqrt{\meanc(x)} = b\leq 1$,
we get
\begin{equation*} 
\arccos \sqrt{\meanm(x)}+ \arccos \sqrt{\meanc(x)} \geq \arccos \sqrt{\sigma_{1}}
\end{equation*}  
in the domain $\meanm(x) \meanc(x) \geq \sigma_{1}$. Applying the trigonometric identity 
$\cos (\alpha - \beta) = \cos \alpha \cos \beta - \sin \alpha \cos \beta$ we finally obtain
the inequality \eqref{equation-uncerataintyangle1b} as
\[\meanc(x)^{\frac{1}{2}} \leq \cos ( \arccos \sqrt{\sigma_{1}} - \arccos \sqrt{\meanm(x)})
=  (\meanm(x) \sigma_{1})^{\frac{1}{2}} + ((1-\meanm(x))
(1-\sigma_{1}))^{\frac{1}{2}}.\]
\end{proof}

\begin{remark}
Note that for $t \in [\sigma_{1},1]$ we have the following inequalities
\[ \sigma_{1} \leq \frac{\sigma_{1}}{t} \leq - t + 1 + \sigma_{1} \leq \gamma_{f,g}(t) \leq 1.  \]
Therefore, in the square $[\sigma_{1},1]^2$ we have the relations
\[ \left\{(t,s) \in [\sigma_{1},1]^2 \ | \ s \geq \gamma_{f,g}(t) \right\} \subset \left\{(t,s) \in [\sigma_{1},1]^2 \ | \ t s \geq \sigma_{1} \right\} \subset [\sigma_{1},1]^2 . \]
\end{remark}

Lemma \ref{lemma-uncertaintyangle} provides a general restriction 
of the set $\mathcal{W}(\mathbf{M}_f, \mathbf{C}_g)$ in the upper right corner 
$[\sigma_{1},1]^2$ of the unit square. By simple reflections, we get an analogous result for the 
other three corners. To simplify the notation we define the corresponding reflection operator $^*$ on the filters $f$ and $g$ as $f^* = 1 - f$ and $g^* = 1 - g$. Further, to distinguish eigenvalues $\sigma_{1}$ for different filters, we use in this part the extended notation $\sigma_1^{(f,g)}$ to denote the largest eigenvalue of the operator $\Sop$. We consider now the following subdomain of the square $[0,1]^2$ (see Figure \ref{Figure-uncertainty} (left)):

\begin{align*}
\mathcal{W}_{\gamma} &:= \left\{(t,s) \in [0,1]^2 \, \left| \; 
\begin{array}{ll}
s \leq \gamma_{f,g}(t) & \text{if}\; t s \geq \sigma_{1}^{(f,g)},\\
1-s \leq \gamma_{f,g^*}(t) & \text{if}\; t (1-s) \geq \sigma_1^{(f,g^*)},\\
s \leq \gamma_{f^*,g}(1-t) & \text{if}\; (1-t) s \geq \sigma_1^{(f^*,g)},\\
1-s \leq \gamma_{f^*,g^*}(1-t) & \text{if}\; (1-t) (1-s) \geq \sigma_1^{(f^*,g^*)}
\end{array}\right.
\right\}.
\end{align*}

Lemma \ref{lemma-uncertaintyangle} now implies the following:

\begin{theorem} \label{Theorem-uncertainty}
The range $\mathcal{W}(\mathbf{M}_f, \mathbf{C}_g)$ is contained in $\mathcal{W}_{\gamma}$.
\end{theorem}

\begin{remark} 
\leavevmode
\begin{enumerate}
\item If $\sigma_{1} = \sigma_{1}^{(f,g)}< 1$ (or, similarly, if $\sigma_1^{(f^*,g)}<1$, $\sigma_1^{(f,g^*)}<1$ or 
$\sigma_1^{(f^*,g^*)}<1$) then Theorem \ref{Theorem-uncertainty} is an uncertainty relation for the operators $\mathbf{M}_f$ and $\mathbf{C}_g$. It states that a signal $x$ on the graph can not be well localized with respect to both operators $\mathbf{M}_f$ and $\mathbf{C}_g$. In particular, the 
vector of mean values $(\meanm(x),\meanc(x))$ can not get close to $(1,1)$. 
\item The uncertainty statement in Theorem \ref{Theorem-uncertainty} can get sharp in the sense that if $\mathbf{M}_f$ and $\mathbf{C}_g$ are two projection operators, then we get equality $\mathcal{W}(\mathbf{M}_f, \mathbf{C}_g) = \mathcal{W}_{\gamma}$ in Theorem \ref{Theorem-uncertainty}. For graphs this is shown in \cite[Theorem 3.1]{TBL2016}. This fact can be interpreted in the following way: among all pairs of positive definite operators with spectral norm $1$
and eigenvalues $\sigma_{1}^{(f,g)}$, $\sigma_1^{(f^*,g)}$, $\sigma_1^{(f,g^*)}$, and $\sigma_1^{(f^*,g^*)}$, a pair of projection operators gives the weakest possible uncertainty relation in Theorem \ref{Theorem-uncertainty}. In other words, pairs of projection operators have the smallest mutual correlation between space- and frequency localization according to this uncertainty relation.  
\end{enumerate}
\end{remark}

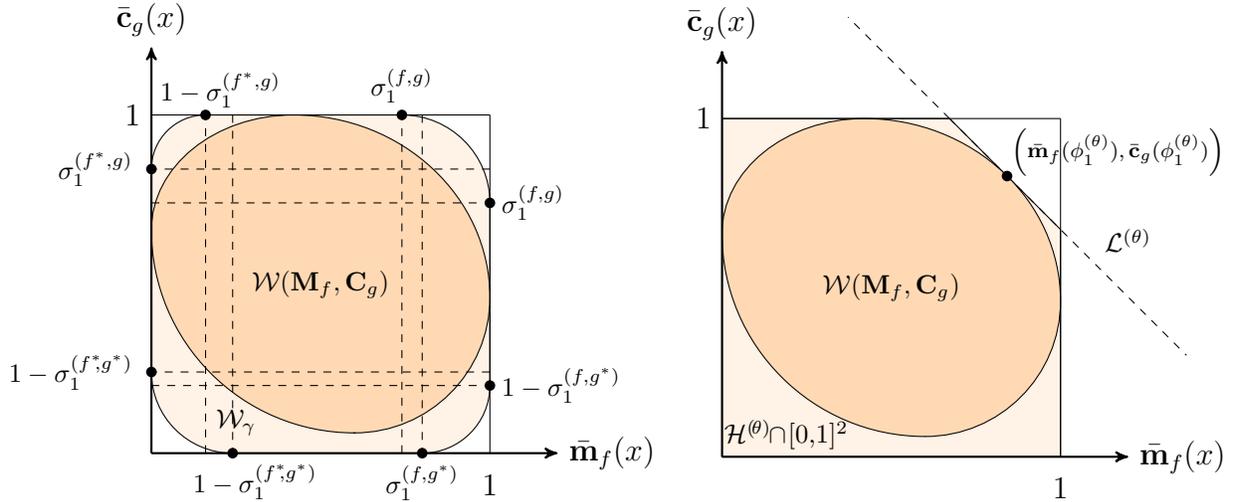
\begin{figure} 
\centering

\begin{tikzpicture}[scale=0.9]

\tikzset{
axis/.style={thick, ->, >=stealth'},
help lines/.style={dashed},
important line/.style={thick},
dot/.style={circle,fill=black,minimum size=4pt,inner sep=0pt,
            outer sep=-1pt}
}

    \def\xmax{3.7}
    \def\xmin{1.2}
    
    \def\ymax{4}
    \def\ymin{0.8}
    
    \filldraw [fill=orange!10!white] 
    (0,\xmin) -- 
    (0,\xmin) arc (180:270: \xmin cm and \xmin cm ) -- 
    (\ymax,0) --
    (\ymax, 0) arc (-90:0: 5 cm - \ymax cm and 5 cm -\ymax cm ) --
    (5,5-\ymax) -- 
    (5, \xmax) arc (0:90: 5 cm - \xmax cm and 5 cm -\xmax cm ) -- 
    (\xmax,5) -- 
    (\ymin,5) arc (90:180: \ymin cm and \ymin cm ) -- 
    (0,5-\ymin)-- cycle;
    
     \filldraw [fill=orange!30!white] 
    (3,0.3) arc (-90:0: 2 cm and 2 cm ) -- 
    (5,2.3) arc (0:90: 2.9 cm and 2.7 cm ) --
    (2.1,5) arc (90:180: 2.1 cm and 1.7 cm ) --
    (0,3.3) arc (180:270: 3 cm and 3 cm ) -- cycle;

    \draw[axis] (0,0)  --  (0,6) node(yline)[above] {$\meanc(x)$};
    \draw[axis] (0,0)  --  (6,0) node(xline)[right] {$\meanm(x)$};
    \draw (5,0) -- (5,5);
    \draw (0,5) -- (5,5);

    \draw (5,0)  node[label=below: $1$]{};
    \draw (0,5) node[left]{$1$};
    
    \draw (2.5,2.5) node {\small $\mathcal{W}(\mathbf{M}_f,\mathbf{C}_g)$};
    \draw (1.25,0.45) node {\small $\mathcal{W}_{\gamma}$};
    
    \small    
    \draw[help lines] (\xmax,0) -- (\xmax,5) node[dot,label=above:$\sigma_{1}^{(f,g)}$]{};
    \draw[help lines] (\xmin,5) -- (\xmin,0) node[dot]{};
    \draw (\xmin,0) node[label=below: $1-\sigma_1^{(f^*\!\!,g^*)}$,xshift=0.3cm,yshift=0.1cm]{};
    \draw[help lines] (5,\xmin) -- (0,\xmin) node[dot]{};
    \draw (0,\xmin) node[label=left: $1-\sigma_1^{(f^*\!\!,g^*)}$]{};
    \draw[help lines] (5,5-\ymin) -- (0,5-\ymin) node[dot]{};
    \draw (0,5-\ymin) node[label=left: $\sigma_1^{(f^*,g)}$]{};
    \draw[help lines] (\ymax,5) -- (\ymax,0) node[dot,label=below:$\sigma_1^{(f,g^*)}$]{};
    \draw[help lines] (\ymin,0) -- (\ymin,5) node[dot]{};
    \draw (\ymin,5) node[label=above: $1-\sigma_1^{(f^*,g)}$,xshift=0.2cm,yshift=-0.15cm]{};
    \draw[help lines] (0,\xmax) -- (5,\xmax) node[dot,label=right:$\sigma_{1}^{(f,g)}$]{};
    \draw[help lines] (0,5-\ymax) -- (5,5-\ymax) node[dot,label=right:$1-\sigma_1^{(f,g^*)}$]{};
     

\end{tikzpicture}
\begin{tikzpicture}[scale=0.9]

\tikzset{
axis/.style={thick, ->, >=stealth'},
help lines/.style={dashed},
important line/.style={thick},
dot/.style={circle,fill=black,minimum size=4pt,inner sep=0pt,
            outer sep=-1pt}
}

    \def\xmax{5-1.64}
    
    \filldraw [fill=orange!10!white] 
    (0,0) -- (5,0) --
    (5, \xmax)  -- 
    (\xmax,5) -- 
    (0,5) -- cycle;
    
     \filldraw [fill=orange!30!white] 
    (3,0.3) arc (-90:0: 2 cm and 2 cm ) -- 
    (5,2.3) arc (0:90: 2.9 cm and 2.7 cm ) --
    (2.1,5) arc (90:180: 2.1 cm and 1.7 cm ) --
    (0,3.3) arc (180:270: 3 cm and 3 cm ) -- cycle;

    \draw[axis] (0,0)  --  (0,6) node(yline)[above] {$\meanc(x)$};
    \draw[axis] (0,0)  --  (6,0) node(xline)[right] {$\meanm(x)$};
    \draw (5,0) -- (5,5);
    \draw (0,5) -- (5,5);
    
    \draw[help lines] (3.36 + 0, 5 + 0) -- (3.36 + 3.5, 5 - 3.5) {};
    \draw[help lines] (3.36 + 0, 5 + 0) -- (3.36 - 1.5, 5 + 1.5) {};

    \draw (5,0) node[label=below: $1$]{};
    \draw (0,5) node[left]{$1$};
    
    \draw (3.36 + 0.85, 5 - 0.85) node[dot]{};
    
    \draw (5.85, 4.5 ) node {\scriptsize $\left(\meanm(\phi_{1}^{(\theta)}),\meanc(\phi_{1}^{(\theta)})\right)$};

    \draw (2.5,2.5) node {\small $\mathcal{W}(\mathbf{M}_f,\mathbf{C}_g)$};
    \draw (0.95,0.33) node {\small $\mathcal{H}^{(\!\theta)}\!\! \cap \! [0,\!1]^2$};
    \draw (6,3.2) node {\small $\mathcal{L}^{(\theta)}$};

\end{tikzpicture}

\caption{Illustration of the uncertainty principles in Theorem \ref{Theorem-uncertainty} (left) and in Theorem \ref{cor-UP-rotatedoperator} (right).}

\label{Figure-uncertainty}

\end{figure}

\subsection{Do we always have uncertainty on graphs?}

In our uncertainty framework, the answer is no. Compared to the real line setting studied in \cite{LandauPollak1961}, the scenario $\sigma_{1} = 1$ is possible on some graphs, as pointed out it \cite{TBL2016}. This implies that in some cases there are signals $x$ satisfying $\meanm(x) = \meanc(x) = 1$, i.e. $x$ is perfectly localized in space and frequency simultaneously. Therefore, we can not expect that every pair of filters $(f,g)$ induces an uncertainty principle on $G$. Nevertheless, in a lot of cases the condition $\sigma_{1} < 1$ can be guaranteed. One of these conditions is the following:

\begin{proposition} \label{prop-sufficient} Let $f$ and $g$ be two filter functions on a graph $G$
satisfying \eqref{eq:propspacefreq}. If the maximal eigenvalue $1$ of $\mathbf{M}_f$ and $\mathbf{C}_g$ is simple, then $\sigma_1 = \|\Sop\| < 1$. 
\end{proposition}

\begin{proof}
Assume that $\sigma_1 = 1$, then by Property \ref{res:sigma} we can find a normalized signal $x$ such
that $1 = \| \mathbf{C}_{\groot} \mathbf{M}_{\froot} x \|$. As the spectral norms 
of $\mathbf{C}_{\groot}$ and $\mathbf{M}_{\froot}$ are one, this is only possible if $x$ is an
eigenvector of $\mathbf{M}_{\froot}$ with respect to the eigenvalue $1$. As $\mathbf{M}_{\froot}$ is a diagonal matrix and the eigenvalue $1$ is simple, we have $x = \pm e_i$  for some $i \in \{1, \ldots, n\}$, i.e. $x$ is up to the factor $\pm$ a canonical basis vector. Therefore, we get $1 = \| \mathbf{C}_{\groot} e_i \|$, i.e. $e_i$ is also an eigenvector of $\mathbf{C}_{\groot}$ with respect to the largest
eigenvalue $1$. As the eigenvalue $1$ of $\mathbf{C}_{\groot}$ is simple this implies that $e_i$ 
corresponds (up to a possible sign) to one of the columns of $\mathbf{U}$, that is, $e_i$ is an eigenvector of the normalized graph Laplacian $\mathbf{L}$. This on the other hand is not possible by the given structure of the adjacency matrix $\mathbf{A}$ as the vertex $v_i$ is connected by at least one edge to another vertex $v_j$. 
\end{proof}

For some graphs, the conditions of Proposition \ref{prop-sufficient} on the filters $f$ and $g$ can not be weakened in order to still guarantee $\sigma_1 < 1$. We give two counterexamples.

\begin{cexample}
\leavevmode
\begin{enumerate}[label = (\arabic*)]
\item (Bipartite graphs) We consider a bipartite graph $G$ with $4$ nodes $\{v_1,v_2,v_3,v_4\}$ and two undirected edges connecting $v_1$ with $v_2$ and $v_3$ with $v_4$. For this graph we obtain the graph Laplacian $\mathbf{L}$ and its spectral decomposition as 
\[ \mathbf{L} = \begin{pmatrix} 1 & - 1 & 0 & 0 \\ -1 & 1 & 0 & 0 \\ 0 & 0 & 1 & -1 \\ 0 & 0 & -1 & 1 \\ \end{pmatrix}, \quad \mathbf{U} = \frac1{\sqrt{2}} \begin{pmatrix} 1 & 0 & 1 & 0 \\ 1 & 0 & -1 & 0 \\ 0 & 1 & 0 & 1 \\ 0 & 1 & 0 & -1 \\ \end{pmatrix}, \quad  
\mathbf{M}_{\lambda} = \begin{pmatrix} 0 & 0 & 0 & 0 \\ 0 & 0 & 0 & 0 \\ 0 & 0 & 2 & 0 \\ 0 & 0 & 0 & 2 \\ \end{pmatrix}.\]
In particular $\lambda_1 = \lambda_2 = 0$ and $\lambda_3 = \lambda_4 = 2$ are double eigenvalues of $\mathbf{L}$ with a corresponding two-dimensional eigenspace. Now, if we choose the filters $f$ and $g$ as
\[ f = (1,0,1,0), \quad \hat{g} = (1,0,1,0),\]
we get
\begin{align*} 
& (\meanm(e_1),\meanc(e_1)) = (1,1), & (\meanm(e_2),\meanc(e_2)) = (0,1), \\
& (\meanm(e_3),\meanc(e_3)) = (1,0), & (\meanm(e_4),\meanc(e_4)) = (0,0).
\end{align*}
The convexity of the numerical range (established in Theorem \ref{th:hausdorff-toeplitz} below) therefore implies that $\mathcal{W}(\mathbf{M}_f, \mathbf{C}_g) = [0,1]^2$, i.e. we encounter no uncertainty in this example. 

In a similar way, we can check for the filters $f = (1,0,0,0)$, $\hat{g} = (1,0,1,0)$ or $f = (1,1,0,0)$, $\hat{g} = (1,0,0,0)$ that the right upper corner $(1,1)$ is contained in $\mathcal{W}(\mathbf{M}_f, \mathbf{C}_g)$, i.e., that $\sigma_1 = 1$. Therefore, in this example also the conditions of Proposition \ref{prop-sufficient} can not be weakened. Similar counterexamples can be constructed on larger bipartite graphs with an even number of nodes. 

\item (Complete graphs) We consider now a complete graph $G$ with $4$ nodes $\{v_1,v_2,v_3,v_4\}$ in which each node is connected to all other nodes by an undirected edge. For this graph we obtain the graph Laplacian $\mathbf{L}$ and its spectral decomposition as 
\[ \mathbf{L} = \frac13 \begin{pmatrix} 3 & - 1 & - 1 & -1 \\ -1 & 3 & -1 & -1 \\ -1 & -1 & 3 & -1 \\ -1 & -1 & -1 & 3 \\ \end{pmatrix}, \quad \mathbf{U} = \frac1{2} \begin{pmatrix} 1 & 1 & \sqrt{2} & 0 \\ 1 & 1 & -\sqrt{2} & 0 \\ 1 & -1 & 0 & \sqrt{2} \\ 1 & -1 & 0 & -\sqrt{2} \\ \end{pmatrix}, \quad  
\mathbf{M}_{\lambda} = \begin{pmatrix} 0 & 0 & 0 & 0 \\ 0 & 4/3 & 0 & 0 \\ 0 & 0 & 4/3 & 0 \\ 0 & 0 & 0 & 4/3 \\ \end{pmatrix}.\]
For the filter functions $f$ and $g$ given by
\[ f = (1,1,0,0), \quad \hat{g} = (0,0,1,0),\]
we get $(\meanm(u_3),\meanc(u_3)) = (1,1)$.
Thus, also in this example of a connected graph the right upper corner $(1,1)$ is contained in $\mathcal{W}(\mathbf{M}_f, \mathbf{C}_g)$ and $\sigma_1 = 1$. In this example it is therefore not possible to weaken the condition for the spatial filter $f$ in Proposition \ref{prop-sufficient}. Also on general complete graphs with $n \geq 3$ nodes a similar counterexample can be constructed. 
\end{enumerate}
\end{cexample}

\section{Computation of Uncertainty principles and the numerical range $\mathcal{W}(\mathbf{M}_f, \mathbf{C}_g)$ }
\label{sec:numericalrange}

For a normalized vector $x \in \Rr^n$ we have
\[ \meanm(x) + i \, \meanc(x) = x^{\intercal} \mathbf{M}_f x + i\, x^{\intercal} \mathbf{C}_g x
= x^{\intercal} (\mathbf{M}_f + i\ \mathbf{C}_g) x.\]
Thus, by identifying the complex numbers $\Cc$ with the plane $\Rr^2$ the admissibility region $\mathcal{W}(\mathbf{M}_f, \mathbf{C}_g)$ can be seen as a part of the numerical range
$\mathcal{W}(\mathbf{M}_f+ i \mathbf{C}_g)$ of the matrix $\mathbf{M}_f+ i \mathbf{C}_g$ given by
\[ \mathcal{W}(\mathbf{M}_f+ i \mathbf{C}_g) := \left\{ \bar{x}^{\intercal} (\mathbf{M}_f+ i \mathbf{C}_g) x \ | \ x \in \Cc^n, \, \|x\|=1 \right\}. \]
The definition of $\mathcal{W}(\mathbf{M}_f,\mathbf{C}_g)$ in \eqref{eq:numericalrange} and of $\mathcal{W}(\mathbf{M}_f+ i \mathbf{C}_g)$ are very similar, the only difference being that
$\mathcal{W}(\mathbf{M}_f+ i \mathbf{C}_g)$ is classically defined in terms of complex-valued vectors $x$. 
In Theorem \ref{th:hausdorff-toeplitz} below, we will see that the two sets coincide if the number of nodes is $n \geq 3$. For this reason, we call the admissibility region $\mathcal{W}(\mathbf{M}_f, \mathbf{C}_g)$ also the numerical range of the pair $(\mathbf{M}_f,\mathbf{C}_g)$. The deep link between the numerical range and uncertainty principles is pointed out in several works, among others in the original work \cite{LandauPollak1961} of Landau and Pollak and the subsequent study in \cite{Lenard1972}. In \cite{Klaja2016}, this link is used to derive uncertainty principles on an interval in terms of general spatial localization measures.

\subsection{Uncertainty principle related to the operator $\Rop$}
\label{sec:propertiesnumericalrange}
The correspondence of $\mathcal{W}(\mathbf{M}_f,\mathbf{C}_g)$ with $\mathcal{W}(\mathbf{M}_f+ i \mathbf{C}_g)$ is important for us, as we can use a broad arsenal of available results for $\mathcal{W}(\mathbf{M}_f+ i \mathbf{C}_g)$ to describe and approximate $\mathcal{W}(\mathbf{M}_f,\mathbf{C}_g)$. A second crucial property for our investigations is
the convexity of $\mathcal{W}(\mathbf{M}_f,\mathbf{C}_g)$. 

\begin{theorem}[Theorem 2.1, 2.2. and Remark 1 in \cite{Brickman}] \label{th:hausdorff-toeplitz} 
\leavevmode \newline
If $n \geq 3$, the set $\mathcal{W}(\mathbf{M}_f, \mathbf{C}_g)$ is convex, compact and corresponds to the numerical range $\mathcal{W}(\mathbf{M}_f+ i \mathbf{C}_g)$. In the case
$n = 2$, the set $\mathcal{W}(\mathbf{M}_f, \mathbf{C}_g)$ corresponds to the elliptical boundary
of $\mathcal{W}(\mathbf{M}_f+ i \mathbf{C}_g)$. 
\end{theorem}

\begin{remark}
The convexity of $\mathcal{W}(\mathbf{M}_f+ i \mathbf{C}_g)$ is the well-known
Hausdorff-Toeplitz Theorem (cf. the original works \cite{Hausdorff1919,Toeplitz1918}
of Hausdorff and Toeplitz, proofs in english are given in \cite[Theorem 1.1-2]{GustafsonRao} or
\cite[Section 1.3]{HornJohnson1991}). Theorem \ref{th:hausdorff-toeplitz}, and, thus, the correspondence of the range $\mathcal{W}(\mathbf{M}_f, \mathbf{C}_g)$ with the classical numerical range $\mathcal{W}(\mathbf{M}_f+ i \mathbf{C}_g)$, is proven in \cite{Brickman}. Actually, 
in \cite{Brickman} this correspondence is shown by first proving the convexity of $\mathcal{W}(\mathbf{M}_f, \mathbf{C}_g)$. A simplified and unified proof for the convexity of the two sets is given in \cite{Au-Yeung1975}. In the exceptional case $n=2$, the set $\mathcal{W}(\mathbf{M}_f, \mathbf{C}_g)$ is an ellipse, a circle or a degenerate ellipse in form of a line segment or a point (cf. \cite{Brickman}). The compactness of $\mathcal{W}(\mathbf{M}_f, \mathbf{C}_g)$ follows from the fact that
$x \to (\meanm(x), \meanc(x))$ is a continuous mapping from the compact unit sphere in $\Rr^n$ onto $\mathcal{W}(\mathbf{M}_f, \mathbf{C}_g)$ (see also \cite[Theorem 5.1-1]{GustafsonRao}).
\end{remark}

Using the convexity of $\mathcal{W}(\mathbf{M}_f, \mathbf{C}_g)$, we derive now further properties that are useful for the formulation of an uncertainty principle as well as for the numerical computation of $\mathcal{W}(\mathbf{M}_f, \mathbf{C}_g)$. The following derivations can already be found in a similar form in the first works \cite{Hausdorff1919,Toeplitz1918} of Hausdorff and Toeplitz for the range $\mathcal{W}(\mathbf{M}_f+ i \mathbf{C}_g)$. The results regarding the approximation of the numerical range with polygons can be found in \cite{Johnson1978} or in \cite[Section 1.5]{HornJohnson1991}.

We first observe that for $(t,s) \in \mathcal{W}(\mathbf{M}_f, \mathbf{C}_g)$ the largest possible value of the coordinate $t$ is attained 
for a normalized eigenvector of $\mathbf{M}_f$ with respect to the largest eigenvalue. By our definition of the space-frequency operator $\mathbf{R}_{f,g}^{(\theta)}$, these are given as $\rho_{1}^{(0)}$ (the largest eigenvalue) and $\phi_{1}^{(0)}$ (a respective eigenvector) of the matrix $\mathbf{R}_{f,g}^{(0)} = \mathbf{M}_f$. In particular, we have 
\[ \rho_{1}^{(0)} = \phi_{1}^{(0)\intercal} \mathbf{M}_f \phi_{1}^{(0)} = \max_{t \in \Rr} \{ t \ | \ (t,s) \in \mathcal{W}(\mathbf{M}_f, \mathbf{C}_g) \}.\]
Therefore the vertical line $\mathcal{L}^{(0)} = \{ (\rho_{1}^{(0)},s) \ | \ s \in \Rr\}$
is a supporting hyperplane for the numerical range $\mathcal{W}(\mathbf{M}_f, \mathbf{C}_g)$ such that 
the half-plane $\{(t,s) \ | \ t \leq \rho_{1}^{(0)}\}$ contains $\mathcal{W}(\mathbf{M}_f, \mathbf{C}_g)$. Further, 
the point $(\phi_{1}^{(0)\intercal} \mathbf{M}_f \phi_{1}^{(0)}, \phi_{1}^{(0)\intercal} \mathbf{C}_g \phi_{1}^{(0)}) \in \mathcal{L}^{(0)}  \cap \mathcal{W}(\mathbf{M}_f, \mathbf{C}_g)$ is on the
boundary of $\mathcal{W}(\mathbf{M}_f, \mathbf{C}_g)$.

In a next step, we consider for $\theta \in [0,2\pi)$ the (clockwise oriented) rotation matrix 
\[ R^{(\theta)} := \begin{pmatrix} \cos \theta & \sin \theta \\ -\sin \theta & \cos \theta \end{pmatrix}.\]
The rotated numerical range  $R^{(\theta)} \, \mathcal{W}(\mathbf{M}_f, \mathbf{C}_g)$ can be written as
\[ R^{(\theta)} \, \mathcal{W}(\mathbf{M}_f, \mathbf{C}_g) = \mathcal{W}(\cos \theta \, \mathbf{M}_f + \sin \theta \, \mathbf{C}_g, - \sin \theta \, \mathbf{M}_f + \cos \theta \, \mathbf{C}_g).\]
Thus, by considering the largest eigenvalue $\rho_{1}^{(\theta)}$ of the symmetric matrix
$\Rop = \cos (\theta) \mathbf{M}_f + \sin (\theta) \mathbf{C}_g$, 
and a corresponding eigenvector $\phi_{1}^{(\theta)}$, the argument above implies that the line
\[\mathcal{L}^{(\theta)} := \{ \rho_{1}^{(\theta)} (\cos \theta, \sin \theta) + \tau(-\sin \theta, \cos \theta) \ | \ \tau \in \Rr\} = \{(t,s) \ | \ \cos (\theta) \, t + \sin (\theta) \, s = \rho_{1}^{(\theta)}\}\]
is a supporting hyperplane of $\mathcal{W}(\mathbf{M}_f, \mathbf{C}_g)$. In particular, $\mathcal{W}(\mathbf{M}_f, \mathbf{C}_g)$ is completely contained in the 
half-plane 
$$ \mathcal{H}^{(\theta)} := \{(t,s) \ | \ \cos (\theta) \, t + \sin (\theta) \, s \leq \rho_{1}^{(\theta)}\}$$ and the point 
$$ p^{(\theta)} := (\phi_{1}^{(\theta)\intercal} \mathbf{M}_f \phi_{1}^{(\theta)}, \phi_{1}^{(\theta)\intercal} \mathbf{C}_g \phi_{1}^{(\theta)}) \in \mathcal{L}^{(\theta)} \cap \mathcal{W}(\mathbf{M}_f, \mathbf{C}_g)$$
lies on the boundary of the numerical range. We summarize this argumentation line in the following uncertainty principle related to the operators $\Rop$ as well as in a characterization of the boundary 
curve of $\mathcal{W}(\mathbf{M}_f, \mathbf{C}_g)$. For the complex-valued numerical range $\mathcal{W}(\mathbf{M}_f + i \mathbf{C}_g)$ this result was originally given in \cite{Toeplitz1918}. We will use a formulation closer to the one given in \cite[Theorem 1 \& 2 \& 3]{Johnson1978}.

\begin{theorem}[Uncertainty principle related to $\Rop$] \label{cor-UP-rotatedoperator}
\leavevmode \newline
For every $0 \leq \theta < 2 \pi$, we have the inclusion
\[\mathcal{W}(\mathbf{M}_f, \mathbf{C}_g) \subseteq [0,1]^2 \cap \mathcal{H}^{(\theta)}, \]
in which the supporting line $\mathcal{L}^{(\theta)}$ intersects the boundary of $\mathcal{W}(\mathbf{M}_f, \mathbf{C}_g)$. On the other hand, for every point $p$ on the boundary of $\mathcal{W}(\mathbf{M}_f, \mathbf{C}_g)$ we have an angle $0 \leq \theta < 2 \pi$ such that $p \in \mathcal{L}^{(\theta)}$. For this angle, we get an eigenvector $\phi_1^{(\theta)}$ (not necessarily unique) corresponding to the largest eigenvalue $\rho_1^{(\theta)}$ of $\Rop$ such that
\[p = (\phi_{1}^{(\theta)\intercal} \mathbf{M}_f \phi_{1}^{(\theta)}, \phi_{1}^{(\theta)\intercal} \mathbf{C}_g \phi_{1}^{(\theta)}).\] 
\end{theorem}

\begin{remark}
\leavevmode
\begin{enumerate}
\item The second statement of Theorem \ref{cor-UP-rotatedoperator} follows from the convexity of $\mathcal{W}(\mathbf{M}_f, \mathbf{C}_g)$ in the case $n \geq 3$. For $n = 2$, we use the fact that $\mathcal{W}(\mathbf{M}_f, \mathbf{C}_g)$ corresponds to the boundary of the convex numerical range $\mathcal{W}(\mathbf{M}_f + i \mathbf{C}_g)$. Both is guaranteed by Theorem \ref{th:hausdorff-toeplitz}. Theorem \ref{cor-UP-rotatedoperator} is illustrated graphically in Figure \ref{Figure-uncertainty} (right).
\item For $\theta = \pi/4$, Theorem \ref{cor-UP-rotatedoperator} implies that
\[ \sqrt{2} \rho_n^{(\pi/4)} \leq \meanm(x) + \meanc(x) \leq \sqrt{2} \rho_1^{(\pi/4)}.\]
Defining, as in Example \ref{sec:examples}.5 the space and frequency operators in the spectral domain instead of in the graph domain, we obtain similarly the inequalities
\[ \sqrt{2} \rho_n^{(\pi/4)} \leq \bar{\mathbf{m}}_{\hat{f}}(\hat{x}) + \bar{\mathbf{c}}_{\hat{g}}(\hat{x}) \leq \sqrt{2} \rho_1^{(\pi/4)},\]
in which $\rho_1^{(\pi/4)}$ and $\rho_n^{(\pi/4)}$ are the largest and the smallest eigenvalue
of the matrix $(\mathbf{M}_{\hat{f}} + \mathbf{C}_{\hat{g}})/\sqrt{2}$. Using the Laplace-Laplace filter described in Example \ref{sec:examples}.5 a variant of this inequality was formulated in \cite[Theorem 4.1]{BenedettoKoprowski2015} as an uncertainty principle on graphs.
\end{enumerate}
\end{remark}

\subsection{Approximation of the numerical range $\mathcal{W}(\mathbf{M}_f,\mathbf{C}_g)$ with polygons}
\label{sec:approximationnumericalrange}

We proceed now one step further and construct polygons based on a set $\Theta = \{\theta_1, \ldots \theta_{K}\} \subset [0,2 \pi)$ of $K \geq 3$ different angles to approximate the numerical range $\mathcal{W}(\mathbf{M}_f, \mathbf{C}_g)$ from the interior as well as from the exterior. Using the 
notation of Section \ref{sec:propertiesnumericalrange}, we define the two $K$-gons
\begin{align*}
\mathcal{P}_{\mathrm{out}}^{(\Theta)}(\mathbf{M}_f, \mathbf{C}_g) &:= \bigcap_{k = 1}^K \mathcal{H}^{(\theta)} \; = \; \bigcap_{k = 1}^K \left\{(t,s) \ | \ \cos (\theta_k) \, t + \sin (\theta_k) \, s \leq \rho_{1}^{(\theta_k)} \right\}, \\ 
\mathcal{P}_{\mathrm{in}}^{(\Theta)}(\mathbf{M}_f, \mathbf{C}_g) &:= \mathrm{conv} \{p^{(\theta_1)}, p^{(\theta_2)}, \ldots p^{(\theta_K)}\}.
\end{align*}
The convexity of the numerical range $\mathcal{W}(\mathbf{M}_f, \mathbf{C}_g)$ (for $n \geq 3$) combined with the statements of Theorem \ref{cor-UP-rotatedoperator} imply the following result.

\begin{theorem}[Theorem 4 in \cite{Johnson1978}] \label{Theorem-uncertaintyconvex}
Let $\Theta = \{\theta_1, \ldots \theta_{K}\} \subset [0,2 \pi)$ be a set of $K \geq 3$ different angles and $n \geq 3$.
Then, 
\[ \mathcal{P}_{\mathrm{in}}^{(\Theta)}(\mathbf{M}_f, \mathbf{C}_g) 
\subseteq \mathcal{W}(\mathbf{M}_f, \mathbf{C}_g) \subseteq 
\mathcal{P}_{\mathrm{out}}^{(\Theta)}(\mathbf{M}_f, \mathbf{C}_g).
\]
\end{theorem}

\begin{remark}
The vertices of the outer polygon $\mathcal{P}_{\mathrm{out}}^{(\Theta)}(\mathbf{M}_f, \mathbf{C}_g)$ can be calculated explicitly. A corresponding formula based on the eigenvalues 
$\rho_{1}^{(\theta_k)}$ is given in \cite{Johnson1978} and adapted to the notation of this article in equation \eqref{eq:outerpolygon}. Note that compared to \cite{Johnson1978}, the orientation of the rotation is reversed.  
\end{remark}

\subsection{Algorithm for the numerical approximation of the numerical range $\mathcal{W}(\mathbf{M}_f,\mathbf{C}_g)$}

Using the version of Theorem \ref{Theorem-uncertaintyconvex} for the range $\mathcal{W}(\mathbf{M}_f + i \mathbf{C}_g)$, two algorithms for the polygonal approximation of the convex set $\mathcal{W}(\mathbf{M}_f + i \mathbf{C}_g)$ (one from the interior, the other from the exterior) were derived in \cite{Johnson1978}. In this article, we can additionally exploit the symmetry of the matrices $\mathbf{M}_f$ and $\mathbf{C}_g$. The resulting purely real-valued method to obtain the polygonal approximations of $\mathcal{W}(\mathbf{M}_f,\mathbf{C}_g)$ is listed in Algorithm \ref{algorithm1}.

\begin{remark} 
\leavevmode
\begin{enumerate} 
\item In Algorithm $1$, we didn't specify a strategy for the selection of the angles $\theta_k$. Such strategies are studied in \cite{Rote1992} in which the resulting method for the approximation of an arbitrary convex set in $\Rr^2$ is called sandwich algorithm (as the boundary of the convex set is sandwiched by an inner and an outer polygon). In \cite{Rote1992}, it is shown that if an adaptive angle bisection is applied then the sandwich algorithm converges quadratically in the number of vertices $K$. 
\item In \cite{AgaskarLu2013}, the sandwich algorithm was applied to approximate a part of the boundary of 
$\mathcal{W}(\mathbf{M}_f, \mathbf{C}_g)$ (denoted as uncertainty curve) in case of the filter pair $(f,g)$ given in Section \ref{sec:examples} (4). Compared to Theorem \ref{cor-UP-rotatedoperator}, a slightly different characterization of the boundary points of $\mathcal{W}(\mathbf{M}_f, \mathbf{C}_g)$ was derived in \cite[Theorem 1]{AgaskarLu2013}. Namely, instead of a rotation angle $\theta$ a slope parameter $\alpha$ was used. Although the characterization with a slope parameter $\alpha$ is elegant, it has the slight disadvantage that the entire boundary of $\mathcal{W}(\mathbf{M}_f, \mathbf{C}_g)$ can not be described with a single parametrization. 
\end{enumerate} 
\end{remark}

\begin{algorithm}[H] \label{algorithm1}

\caption{Calculation of interior and exterior approximations to $\mathcal{W}(\mathbf{M}_f,\mathbf{C}_g)$}

\begin{multicols}{2}

\vspace{4mm}

\KwIn{The matrices $\mathbf{M}_f$, $\mathbf{C}_g$, 
      the angles $0 \leq\! \theta_1\!< \!\theta_2\!< \!\cdots\! < \!\theta_K\! <  2 \pi$, with $K \geq 3$. Set $\theta_{0} = \theta_K$.   
}

\vspace{2mm}

\For{$k \in \{ 1,2,\ldots, K\}$}{
  Create $\mathbf{R}_{f,g}^{(\theta_k)} = \cos (\theta_k) \mathbf{M}_f + \sin (\theta_k) \mathbf{C}_g$ \;
  
  Calculate normalized eigenvector $\phi_{1}^{(\theta_k)}$ for the maximal eigenvalue $\rho_{1}^{(\theta_k)}$\;
  
  Create the boundary point $$p^{(\theta_k)} = \left(\phi_{1}^{(\theta_k)\intercal} \mathbf{M}_f \phi_{1}^{(\theta_k)}, \phi_{1}^{(\theta_k)\intercal} \mathbf{C}_g \phi_{1}^{(\theta_k)}\right).$$      
  }
\textbf{Generate} the interior polygon $$\mathcal{P}_{\mathrm{in}}^{(\Theta)}(\mathbf{M}_f, \mathbf{C}_g) = \mathrm{conv} \{p^{(\theta_1)}, \ldots p^{(\theta_K)}\}$$ as an approximation to $\mathcal{W}(\mathbf{M}_f,\mathbf{C}_g)$.

\vspace{2mm}

\For{$k \in \{ 1,2,\ldots, K\}$}{
Create the outer vertex $q^{(\theta_k)}$ as
\begin{equation} \label{eq:outerpolygon}
\hspace{-4mm} \scriptsize q^{(\theta_k)} = R^{(-\theta_k)} \left(\rho_{1}^{(\theta_k)},
\frac{\rho_{1}^{(\theta_k)}\cos(\theta_{k}-\theta_{k-1})-\rho_{1}^{(\theta_{k-1})}}{\sin(\theta_{k}-\theta_{k-1})} \right).
\end{equation}    
  }

\textbf{Generate} $\mathcal{P}_{\mathrm{out}}^{(\Theta)}(\mathbf{M}_f, \mathbf{C}_g) = \mathrm{conv} \{q^{(\theta_1)}, \ldots q^{(\theta_K)}\}$ as a polygon exterior to $\mathcal{W}(\mathbf{M}_f,\mathbf{C}_g)$.

\par\medskip 

\begin{minipage}{\linewidth} \centering
\includegraphics[width= 0.7\linewidth]{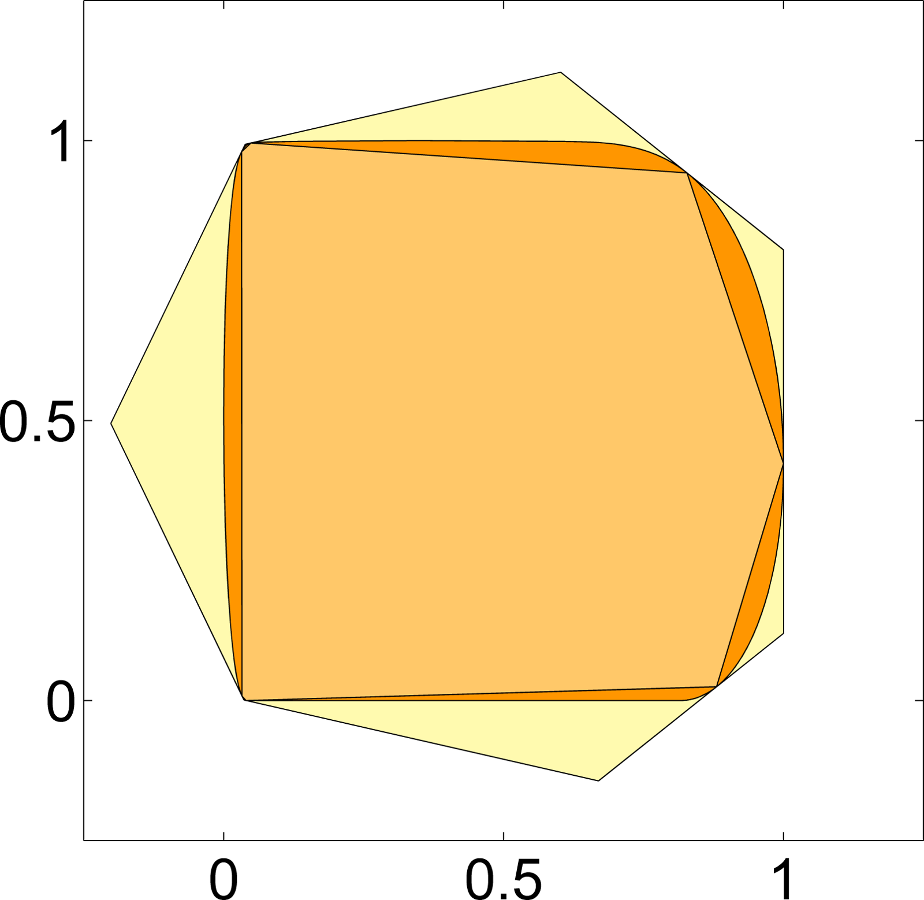}
\end{minipage}
\par\medskip 
Fig. Alg. 1: Interior and exterior approximation of the numerical range $\mathcal{W}(\mathbf{M}_f,\mathbf{C}_g)$ based on Algorithm 1 using an interior and an exterior polygon with $K = 7$ vertices.

\end{multicols}

\end{algorithm}

\section{Error estimates for space-frequency localized signals} \label{sec:errorestimates}

The orthogonal basis of eigenvectors $\{\psi_1, \ldots, \psi_{n}\}$ and $\{\phi_1^{(\theta)}, \ldots, \phi_{n}^{(\theta)}\}$ of the matrices $\Sop$ and $\Rop$ are natural candidates to decompose a signal $x$ on $G$ into single space-frequency components. In particular, we can expand every signal $x$ as
\[x = \sum_{k = 1}^n (\psi_k^{\intercal} x)\,  \psi_k \quad \text{and} \quad x = \sum_{k = 1}^n (\phi_k^{(\theta)\intercal} x) \, \phi_k^{(\theta)},\]
with the coefficients $\psi_k^{\intercal} x$ and $\phi_k^{(\theta)\intercal} x$ giving information about the space-frequency localization of $x$. If the signal $x$ itself is space-frequency localized with respect to the operators $\Sop$ or $\Rop$, or if the variance terms 
\begin{align*}
\vvar[\Sop](x) & :=  \frac{x^{\intercal}(\mathbf{S}_{f,g} - \means(x))^2 x}{\|x\|^2}, \quad \vvar[\Rop](x) := \frac{x^{\intercal}(\Rop - \meanr(x))^2 x}{\|x\|^2} 
\end{align*}
are small, we can approximate the signal $x$ well with only a few eigenvectors. This is specified in the following result. 

\begin{theorem} \label{theorem-localizedapproximation}
Let $s < \sigma_1$ and $r < \rho_1^{(\theta)}$. For a signal $x$ on $G$, we have the inequalities
\begin{align}
& \left\| x - \sum_{k:\, \sigma_k \geq s} ( \psi_k^{\intercal} x ) \psi_k \right\|^2 \leq \frac{\sigma_1-\means(x)}{\sigma_1 - s}
\|x\|^2, \quad \left\| x - \sum_{k:\, \rho_k^{(\theta)} \geq r} ( \phi_k^{(\theta)\intercal} x ) \phi_k^{(\theta)}\right\|^2 \leq \frac{\rho_1^{(\theta)}-\meanr(x)}{\rho_1^{(\theta)} - r}
\|x\|^2.\label{equation-errorbound1}
\end{align}
Further, for $a > 0$, define the intervals $I_{\mathbf{s},a} = [\means(x)-a, \means(x) + a]$, and $I_{\mathbf{r},a} = [\meanr(x)-a, \meanr(x) + a]$. Then, we get the error bounds
\begin{align}
& \left\| x - \sum_{k: \, \sigma_k \in I_{\mathbf{s},a}} ( \psi_k^{\intercal} x ) \psi_k\right\|^2 \leq \frac{\vvar[\Sop](x)}{a^2} \|x\|^2, \quad \left\| x - \sum_{k: \, \rho_k^{(\theta)} \in I_{\mathbf{r},a}} ( \phi_k^{(\theta)\intercal} x ) \phi_k^{(\theta)}\right\|^2 \leq \frac{\vvar[\Rop](x)}{a^2} \|x\|^2.  \label{equation-errorbound2}
\end{align}
\end{theorem}

\begin{proof}
We provide the proof only for the space-frequency analysis related to the operator $\Sop$. For
$\Rop$ the argumentation line is identical. 

For a signal $x$ on $G$, the orthonormality of the eigenbasis $\{\psi_1, \ldots \psi_n\}$ gives
\begin{align*}
& \left\| x - \sum_{k: \, \sigma_k \geq s} ( \psi_k^{\intercal} x ) \psi_k\right\|^2 = 
\sum_{k:\, \sigma_k < s} ( \psi_k^{\intercal} x )^2 \leq
\frac1{\sigma_1-s} \sum_{k:\, \sigma_k < s} ( \psi_k^{\intercal} x )^2 (\sigma_1-\sigma_k) \leq 
\frac1{\sigma_1-s} \sum_{k=1}^n ( \psi_k^{\intercal} x )^2 (\sigma_1-\sigma_k)
\end{align*}
Since, $\|x\|^2 = \sum_{k=1}^n ( \psi_k^{\intercal} x )^2$ (Pythagoras) and
 $\sum_{k = 1}^{n} \sigma_k ( \psi_k^{\intercal} x )^2 = \means(x) \|x\|^2$ (spectral decomposition of 
 $\mathbf{S}_{f,g}$), we get the inequality \eqref{equation-errorbound1}. 
Similarly, we can prove the bound in \eqref{equation-errorbound2}. Namely, we have
\begin{align*}
&  \left\| x - \sum_{k: \, \sigma_k \in I_{\mathbf{s},a}} ( \psi_k^{\intercal} x ) \psi_k\right\|^2
= \sum_{k:\,\sigma_k \in \Rr \setminus I_{\mathbf{s},a}} ( \psi_k^{\intercal} x )^2 
 \leq \frac{1}{a^2} \sum_{k:\,\sigma_k \in \Rr \setminus I_{\mathbf{s},a}} ( \psi_k^{\intercal} x )^2 (\means(x) - \sigma_k)^2 \\
& \quad \leq \frac{1}{a^2} \sum_{k = 1}^{n} ( \psi_k^{\intercal} x )^2 (\means(x) - \sigma_k)^2 = \frac{\vvar[\Sop](x)}{a^2} \|x\|^2.
\end{align*}
This completes the proof of \eqref{equation-errorbound2} for the operator $\Sop$.
\end{proof}

\begin{remark}
For a normalized signal $x$ on $G$ with $\|x\|=1$, the vector $\mu(x) = 
(\mu_1(x), \ldots, \mu_n(x))$ given by $\mu_k(x) = (\psi_k^{\intercal} x )^2$ can be considered
as a probability distribution on the spectrum of $\mathbf{S}_{f,g}$ (similarly also for the operator
$\Rop$). The two inequalities \eqref{equation-errorbound1} and \eqref{equation-errorbound2} stated in Theorem \ref{theorem-localizedapproximation} can therefore be seen as variants of the Markov and the Chebyshev inequality for a $\mu(x)$-distributed random variable, (see \cite[p. 114]{Papoulis}). For orthogonal polynomials on the interval $[-1,1]$, similar error estimates were derived in \cite{erb2013}.
\end{remark}

\section{Shapes of uncertainty - Examples and Illustrations} \label{sec:shapeexamples}

As a final part of this work, we want to study and illustrate the uncertainty regions for concrete filter pairs $(f,g)$. Further, we want to analyze the effects of the different filter pairs on the space-frequency localization on graphs. For this, we conduct several numerical experiments on two explicit graphs.

\subsection{Experimental setup for graphs and filters}

\begin{figure}[htbp]
	\centering
	\includegraphics[width= 1\textwidth]{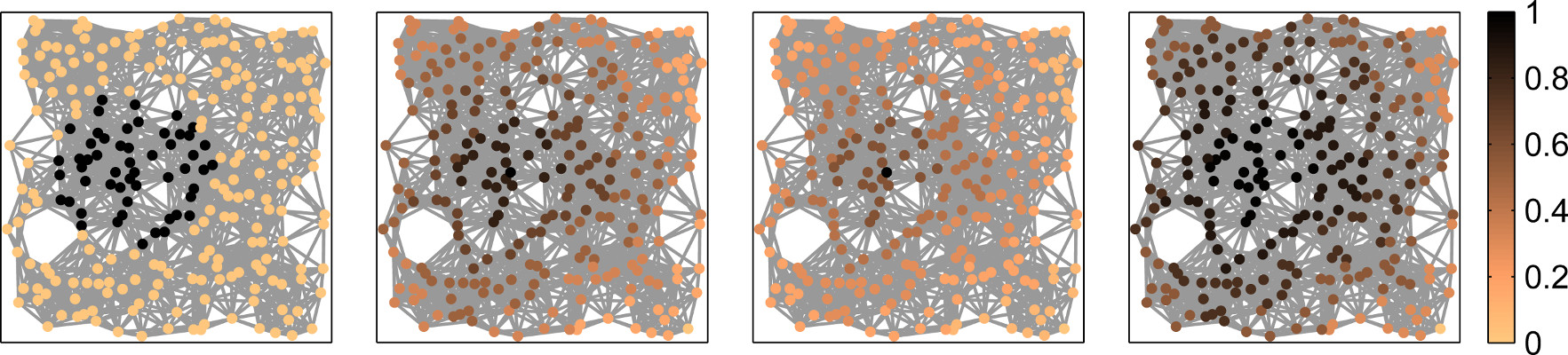}
	\caption{Experimental setup on the sensor network $G_1$. The spatial filters $f_1$, $f_2$, $f_3$ and $f_4$ described in Section \ref{sec:experimentfilters} are plotted from left to right.}
	\label{fig:2}
\end{figure}

\subsubsection{The graphs}
As undirected and unweighted test graphs, we consider point clouds in $\Rr^2$ in which two nodes $v_1$ and $v_2$ get connected if the euclidean distance satisfies $|v_1 - v_2| \leq R$ for some chosen radius $R > 0$. In particular, we study the following two settings:
\begin{enumerate}[label= (\arabic*)]
\item $G_1$ is a sensor network with $n_1 = 253$ random nodes in the square $[0,1]^2$. With the radius $R = 1/6$, we obtain a graph with $2369$ edges. $G_1$ is illustrated in Figure \ref{fig:2}.
\item The node set of $G_2$ is a reduced point cloud taken from the Stanford bunny (Source: Stanford University Computer Graphics Laboratory). It contains $n_2 = 900$ nodes projected in the $xy$-plane. Choosing as radius $R = 0.01$ we obtain the graph $G_2$ with $7325$ edges. The Stanford bunny $G_2$ is illustrated in Figure \ref{fig:3}.  
\end{enumerate}  

\begin{figure}[htbp]
	\centering
	\includegraphics[width= 1\textwidth]{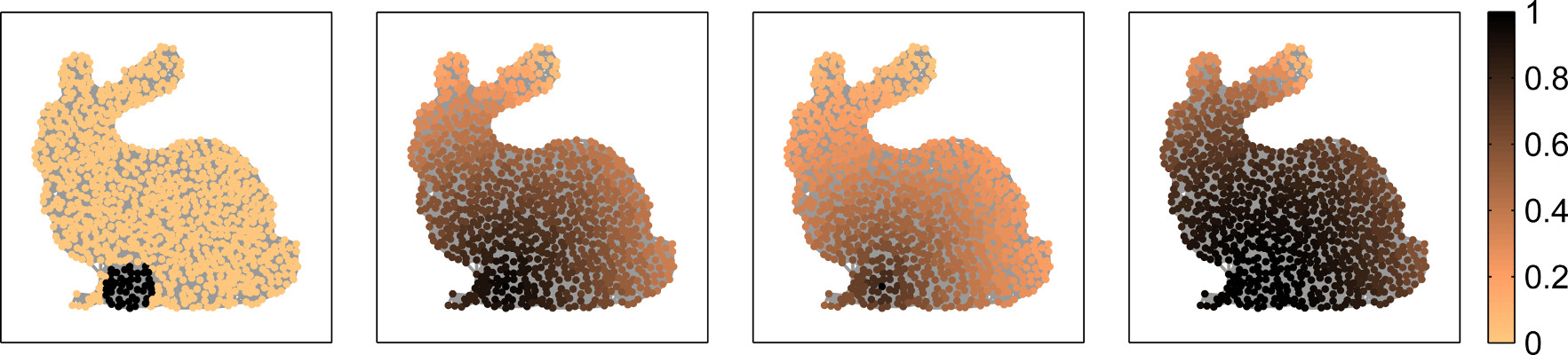}
	\caption{Experimental setup on the Stanford bunny $G_2$. The spatial filters $f_1$, $f_2$, $f_3$ and $f_4$ described in Section \ref{sec:experimentfilters} are plotted from left to right.}
	\label{fig:3}
\end{figure}

\subsubsection{The space and frequency filters} \label{sec:experimentfilters}
We test four different filter pairs:
\begin{enumerate}[label= (\arabic*)]
\item $(f_1,g_1)$ is a \emph{projection-projection} pair as described in Section \ref{sec:examples} (1). It corresponds to the space-frequency setting studied in \cite{TBL2016}. For the spatial filter $f_1 = \chi_{\mathcal{A}}$, we choose the circular set $\mathcal{A} = \{v \in V \ | \ |v - w| \leq r\}$, i.e. $\mathcal{A}$ consists of all nodes of the point cloud $V$ that are within an euclidean distance $r$ to the central node $w$. The matrix $\mathbf{M}_{f_1}$ is then the orthogonal projection onto the 
signals supported in $\mathcal{A}$. For $G_1$, we choose $r = 0.25$, for the bunny $G_2$ we take $r = 0.015$. 

In the spectral domain, we use the filter $\hat{g} = \chi_{\mathcal{B}}$ with $\mathcal{B} = \{u_1, \ldots, u_N\} \subset \hat{G}$ and $N < n$, i.e., $\mathbf{C}_{g_1}$ is the 
orthogonal projection onto the bandlimited signals spanned by the basis $\mathcal{B}$. For the graph $G_1$, we use as bandwidth $N = 100$, for the bunny $G_2$ we take $N = 200$.

\item $(f_2,g_2)$ is a \emph{distance-projection} pair as defined in Section \ref{sec:examples} (2). The spatial filter $f_2$ is defined as
$f_2 = 1 - \distw /\distwmax$, where $\distw(v)$ is the number of edges of the shortest path connecting $w$ with $v \in V$. To compare $f_2$ with $f_1$, we use for both filters the same central node $w$. Further, $g_2$ coincides with the projection filter $g_1 = \chi_{\mathcal{B}}$ described above in (1). 
\item $(f_3,g_3)$ is a \emph{modified distance-projection} pair from Section \ref{sec:examples} (3). The two filters $f$ and $g$ are given for $\alpha > 0$, $\beta > 0$, as 
\begin{equation*} 
f_3 = 1 - \ts \left(\frac{\distw}{\distwmax}\right)^{\alpha}, \quad \hat{g}_3 = \chi_{\mathcal{B}} \odot \left( 1 - \ts \left(\frac{\lambda}{2}\right)^{\beta}\right).
\end{equation*}
Here, $x^{\alpha}$ is defined as $x^{\alpha} = (x_1^\alpha, \ldots, x_n^\alpha )$. The set $\mathcal{B}$ is the same as for the filters $g_1 = g_2$. In our experiments we choose $\alpha = 1/2$ and $\beta = 2$. 
\item $(f_4,g_4)$ is the \emph{distance-Laplace} pair discussed in Section \ref{sec:examples} (4) and a variant of the uncertainty setting studied in \cite{AgaskarLu2013}. This pair is given as 
\[f_4 = 1 - \ts \left(\frac{\distw}{\distwmax}\right)^{2}, \quad \hat{g}_4 = 1 - \lambda/2. \]
In particular, the spatial filter $f_4$ corresponds to the filter $f_3$ with the parameter $\alpha = 2$.  
\end{enumerate}

\subsection{Shapes of uncertainty and space-frequency localization of eigenvectors}

\subsubsection{Description}
As a first experiment, we apply Algorithm 1 and plot the numerical ranges $\mathcal{W}(\mathbf{M}_f,\mathbf{C}_g)$ of the four filter pairs $(f_1,g_1)$, $(f_2,g_2)$, $(f_3,g_3)$, and $(f_4,g_4)$ on 
the two test graphs $G_1$ and $G_2$. Further, we calculate the space-frequency localization 
of the eigenvectors of the matrices $\Sop$ and $\Rop$, $\theta = 9 \pi /20$, inside $\mathcal{W}(\mathbf{M}_f,\mathbf{C}_g)$. 
The corresponding results are illustrated in Figure \ref{fig:shape1}
and Figure \ref{fig:shape2}. As an additional analysis tool, we display in Figure \ref{fig:decaysigma} the decay of the eigenvalues of $\Sop$ and $\Rop$.

\begin{figure}[htbp]
	\centering
	\includegraphics[width= 1\textwidth]{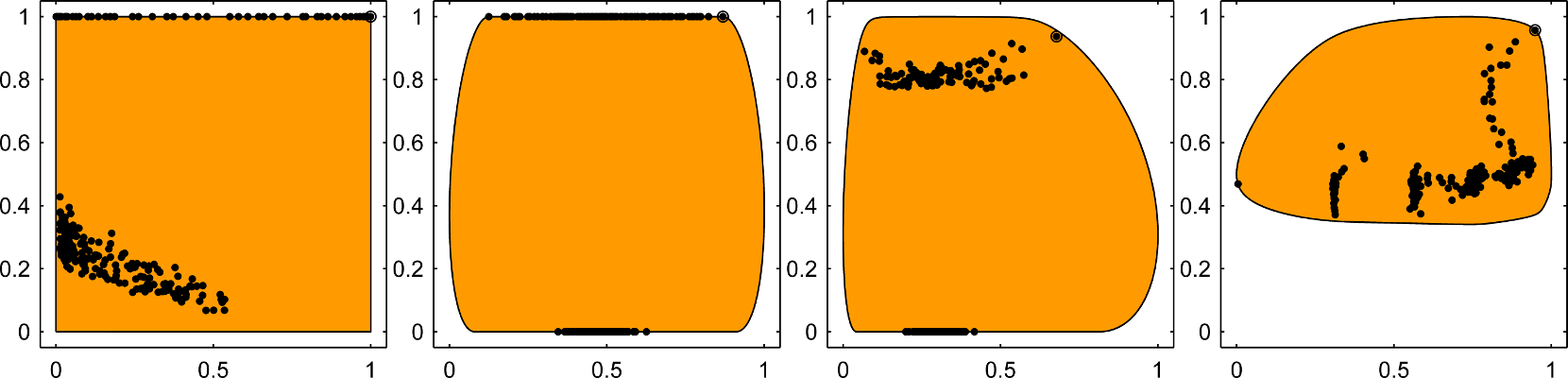}
	\caption{The numerical range $\mathcal{W}(\mathbf{M}_f,\mathbf{C}_g)$ for the filter pairs $(f_1,g_1)$, $(f_2,g_2)$, $(f_3,g_3)$, $(f_4,g_4)$ on the sensor graph $G_1$ (from left to right).
	The black dots represent the position $\left(\meanm(\psi_{k}),\meanc(\psi_{k})\right)$ of the eigenvectors of the operator $\Sop$. The ringed black dot indicates the position of $\psi_{1}$.}
	\label{fig:shape1}
\end{figure}

\begin{figure}[htbp]
	\centering
	\includegraphics[width= 1\textwidth]{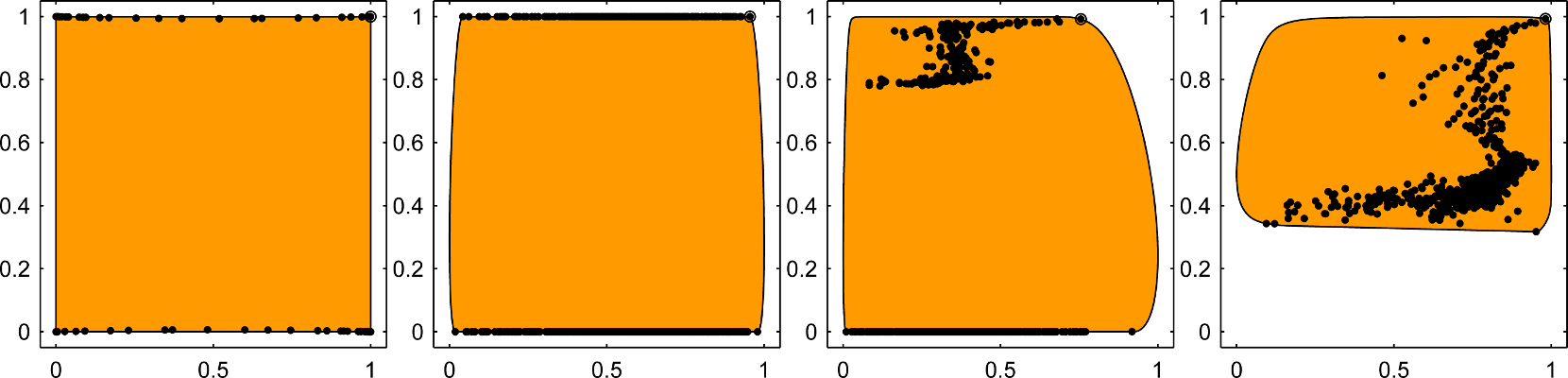}
	\caption{Comparison of $\mathcal{W}(\mathbf{M}_f,\mathbf{C}_g)$ for the filter pairs $(f_1,g_1)$, $(f_2,g_2)$, $(f_3,g_3)$, $(f_4,g_4)$ on the graph $G_2$ (from left to right).
	The black dots represent the location $\left(\meanm(\phi_{k}^{(\theta)}),\meanc(\phi_{k}^{(\theta)})\right)$ of the eigenvectors of the operator $\Rop$ with $\theta = 9 \pi /20$. The ringed black dot indicates the position of $\phi_{1}^{(\theta)}$.}
	\label{fig:shape2}
\end{figure}

\begin{figure}[htbp]
	\centering
	\includegraphics[width= 0.33\textwidth]{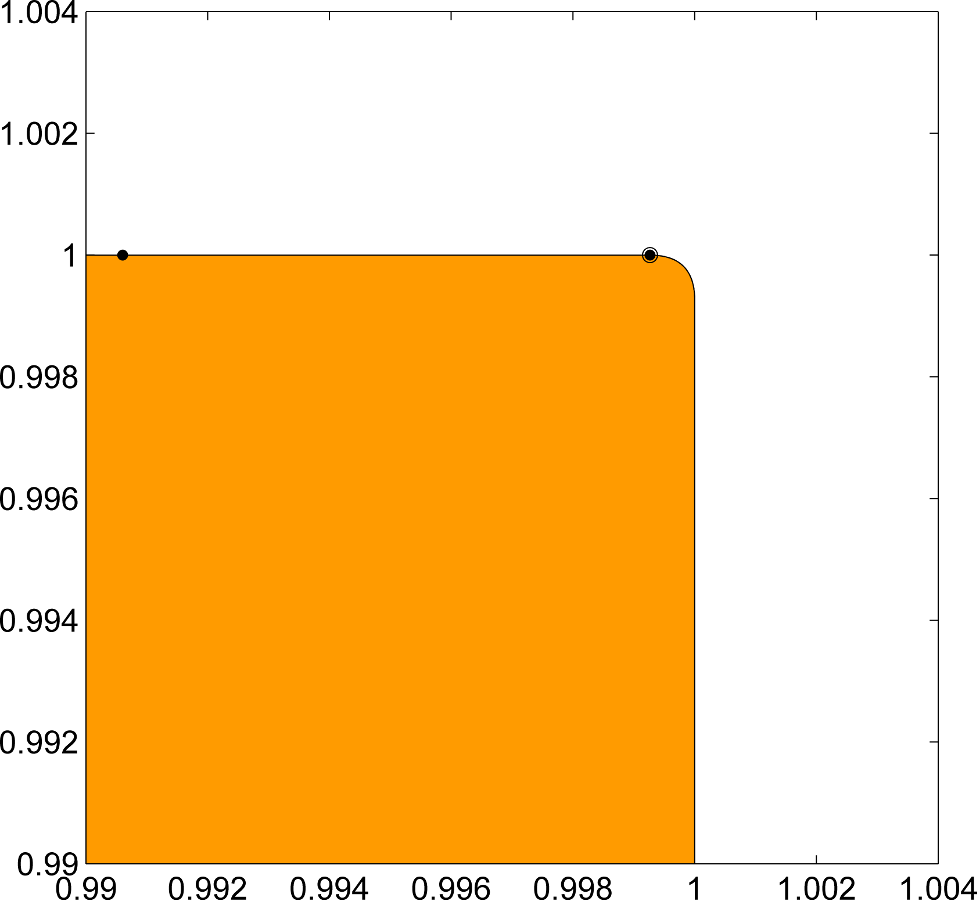}
	\includegraphics[width= 0.32\textwidth]{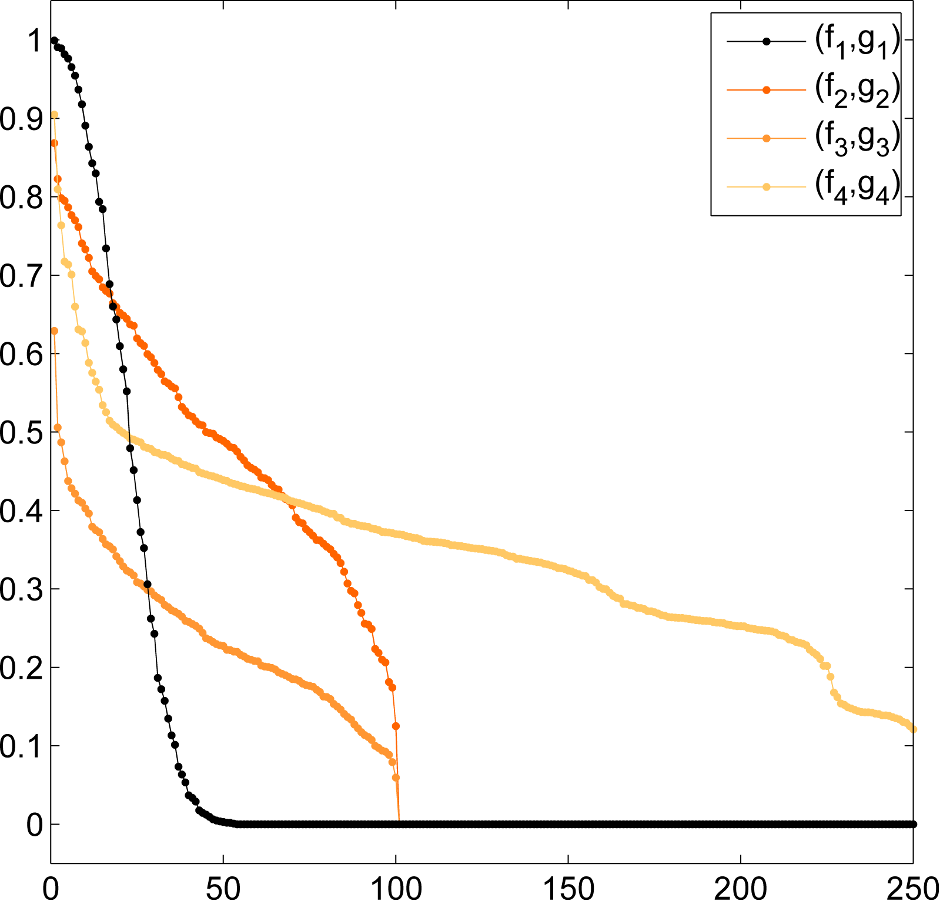}
	\includegraphics[width= 0.32\textwidth]{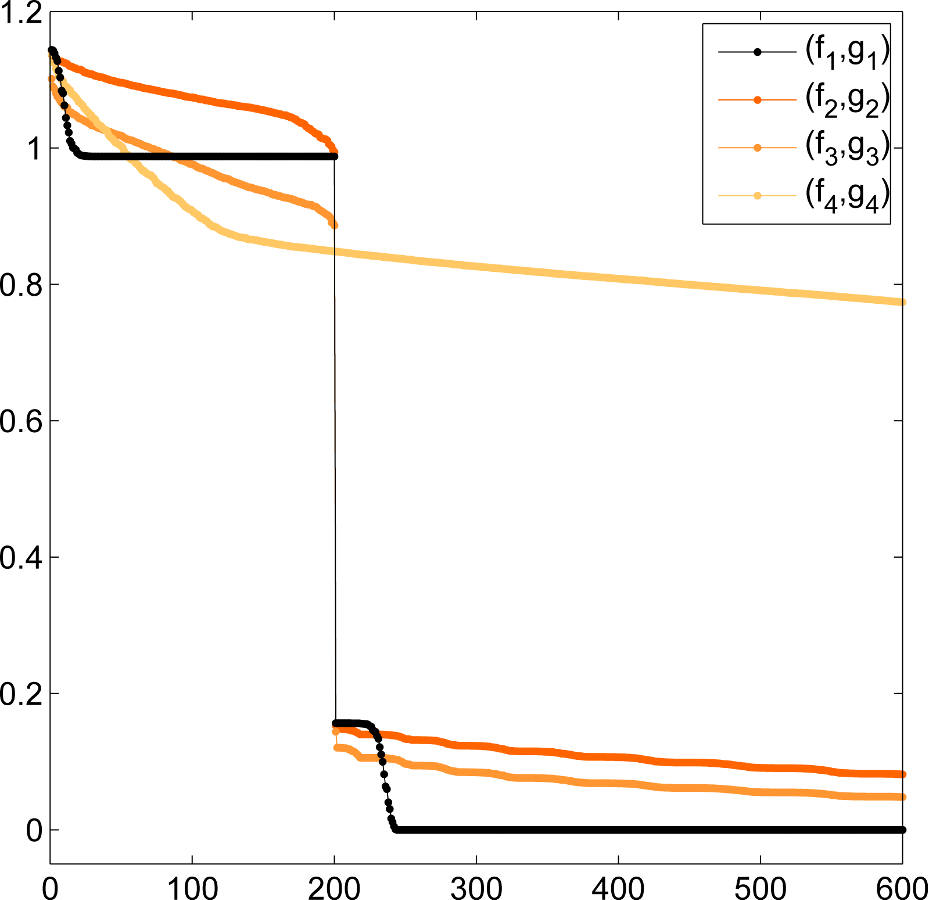}
	\caption{Left: Zoom of the upper right corner of Figure \ref{fig:shape1} (left). Middle: Decay of the eigenvalues $\sigma_k$ of $\Sop$ for the filter pairs $(f_1,g_1)$, $(f_2,g_2)$, $(f_3,g_3)$, $(f_4,g_4)$ on $G_1$. Right: Decay of the eigenvalues $\rho_k^{(\theta)}$ of $\Rop$ ($\theta = 9 \pi /20$) for the filter pairs $(f_1,g_1)$, $(f_2,g_2)$, $(f_3,g_3)$, $(f_4,g_4)$ on the graph $G_2$.}
	\label{fig:decaysigma}
\end{figure}

\subsubsection{Discussion of the shapes} 

From the shape of the uncertainty curves it is possible to extract qualitative information about the applied filter functions, and in case of $(f_4,g_4)$ also about the underlying graph $G$. All four filter pairs display an uncertainty, the projection filter pair $(f_1,g_1)$ giving the largest admissibility region $\mathcal{W}(\mathbf{M}_{f_1},\mathbf{C}_{g_1})$, or in other words, the weakest uncertainty relation. That $\mathcal{W}(\mathbf{M}_{f_1},\mathbf{C}_{g_1})$ describes in fact an uncertainty relation is only visible by a proper zoom, as displayed in Figure \ref{fig:decaysigma} (left). 

The parameter $\alpha>0$ of the modified distance filter $f_3$ has a visible impact on the shape of the uncertainty curve close to $(1,1)$. While decreasing the parameter $\alpha$ results in an uncertainty curve distant to the point $(1,1)$, increasing $\alpha$ has the opposite effect. The spectral filters $g_1$, $g_2$ and $g_3$ are all three bandlimiting filters. This is visible in the first three illustrations of Figure \ref{fig:shape1} and Figure \ref{fig:shape2} as the lower boundary of the numerical range intersects the axis $s = 0$. The fourth filter $g_4$ contains spectral information of the graph. In Figure \ref{fig:shape1} and \ref{fig:shape2} (right) we see that the operator $\mathbf{C}_{g_4}$ is invertible, and, thus that the largest eigenvalue of the graph Laplacian certainly satisfies $\lambda_n < 2$. 

\subsubsection{Discussion of the space-frequency localization of the eigenvectors of $\Sop$ and $\Rop$} 
We first have a look at the decay of the eigenvalues of $\Sop$ and $\Rop$ in Figure \ref{fig:decaysigma} (middle) and (left). The bandlimiting behavior of the spectral filters $g_1$, $g_2$ and $g_3$ is visible by the jumps of the eigenvalues at the bandwidth $N$, whereas for $g_4$ we see a smooth decay of the eigenvalues. For the projection-projection pair $(f_1,g_1)$ an earlier drop of the eigenvalues is visible in case of the operator $\Sop$ and a clustering at the values $0$, $\cos \theta$, $\sin \theta$ and $1$ in case of the operator $\Rop$. The distance filters $f_2$ and $f_3$ on the other hand provide smoothly decaying eigenvalues until the rapid drop at $N$. 

The bandlimiting property of the filters $g_1$, $g_2$ and $g_3$ is also visible in the space-frequency locations of the eigenvectors of $\Sop$ and $\Rop$ shown in Figure \ref{fig:shape1} and \ref{fig:shape2}. For these filters, we see a clear separation between 
bandlimited eigenvectors in the range and the eigenvectors spanning the kernel of $\Sop$ and $\Rop$, respectively. For the pairs $(f_1,g_1)$ and $(f_3,g_3)$ additional effects are visible as $f_1$ is a projection filter (enlarging the kernel of $\Sop$ and $\Rop$) and as $\hat{g}_3$ contains an additional smoothing factor $\hat{g}^{(\beta)}$. For the filter pair $(f_4,g_4)$ such a separation is not visible.

\subsection{Space localization of bandlimited signals for distance-projection filters} 

\begin{figure}[htbp]
	\centering
	\includegraphics[width= 1\textwidth]{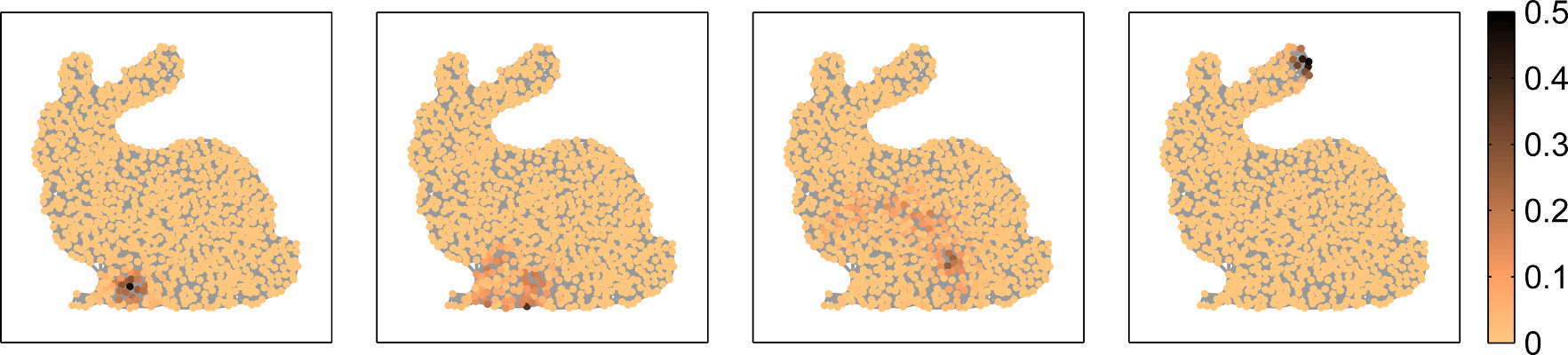} 
	\caption{The eigenvectors $\psi_1$ $\psi_{10}$, $\psi_{50}$ and $\psi_{200}$ of $\mathbf{S}_{f_2,g_2}$ on $G_2$ (from left to right).}
	\label{fig:ringeling}
\end{figure}

In case of the distance-projection pair $(f_2,g_2)$ further interesting effects are visible in the space-frequency behavior of the eigendecomposition of the operator $\Sop$. In the example given in Figure \ref{fig:shape1} (middle left), we observe that the frequency measure of an eigenvector $\psi_k$ is either $\meanc(\psi_k) = 1$ (i.e. the eigenvector $\psi_k$ is bandlimited) or $\meanc(\psi_k) = 0$ (if $k > N$, i.e. the support of $\hat{\psi}_k$ is outside of $\mathcal{B}$). We can further order the bandlimited eigenvectors of $\Sop$ with respect to their spatial localization $\meanm$. This corresponds to the natural ordering of the bandlimited eigenvectors $\psi_k$ with respect to the decreasing eigenvalues of $\Sop$. In particular, the optimally space-localized eigenvector with respect to the localization measure $\meanm$ inside the band $\mathcal{B}$ is $\psi_1$, the least space-localized is the eigenvector $\psi_N$. For the distance filter $f_2$ on the graph $G_2$, different bandlimited eigenvectors $\psi_k$ are illustrated in Figure \ref{fig:ringeling}. It gets visible that the eigenvectors $\psi_k$ are localized on $G_2$ in a ring with a certain graph distance to the center node $w$. This distance is linked to the index $k$.

\subsection{Space-frequency behavior of the optimally localized eigenvectors}

\begin{figure}[htbp]
	\centering
	\includegraphics[width= 1\textwidth]{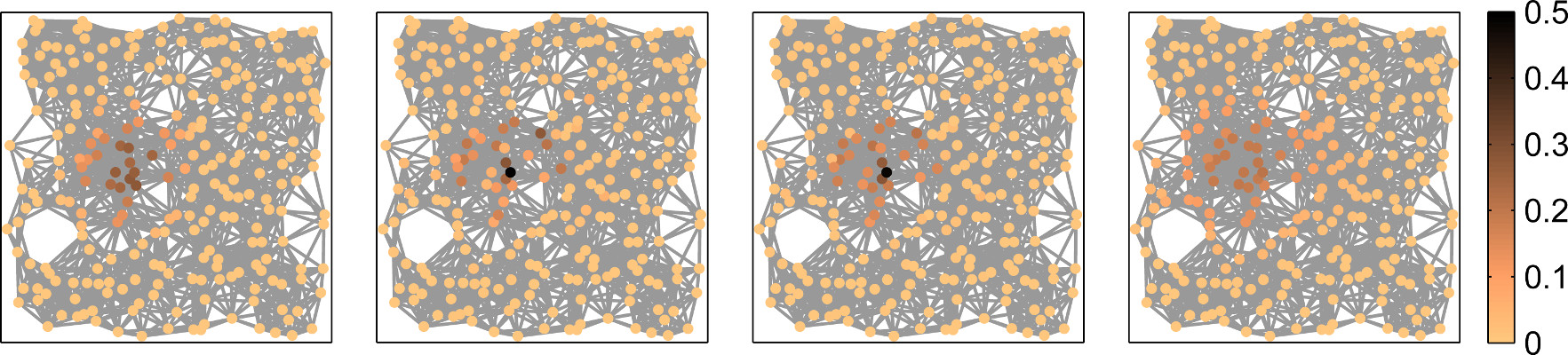} \\[1mm]
	\includegraphics[width= 1\textwidth]{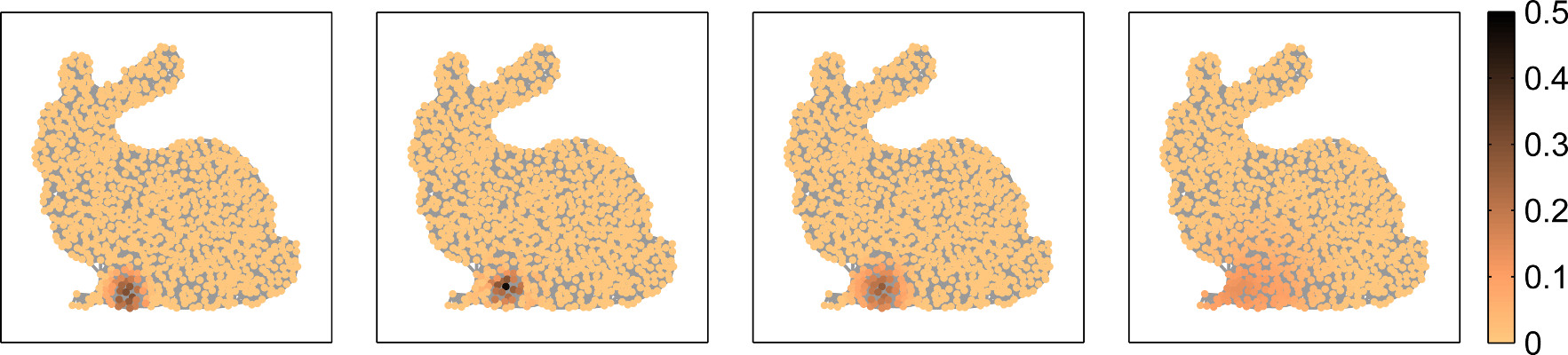}
	\caption{Top row: the eigenvectors $\psi_1$ of the operator $\Sop$ for the graph $G_1$ and the filter pairs $(f_1,g_1)$, $(f_2,g_2)$, $(f_3,g_3)$, and $(f_4,g_4)$ (from left to right). \newline
	Bottom row: the eigenvectors $\phi_1^{(\theta)}$ of the operator $\Rop$ with $\theta = \frac{9}{20} \pi$ for the graph $G_2$ and the filter pairs $(f_1,g_1)$, $(f_2,g_2)$, $(f_3,g_3)$, and $(f_4,g_4)$ (from left to right).}
	\label{fig:optimaleigspace}
\end{figure}

\begin{figure}[htbp]
	\centering
	\includegraphics[width= 1\textwidth]{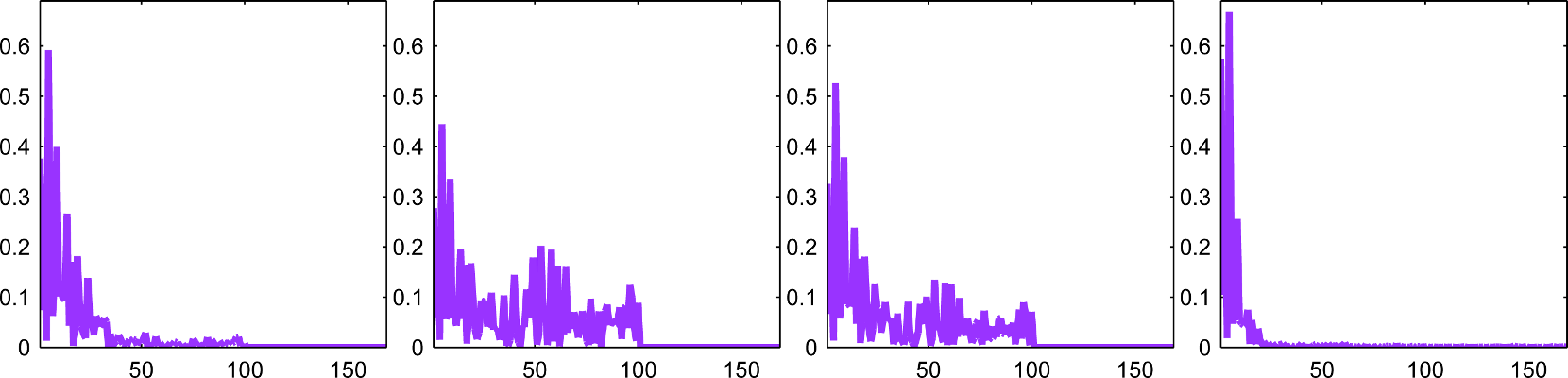} \\[1mm]
	\includegraphics[width= 1\textwidth]{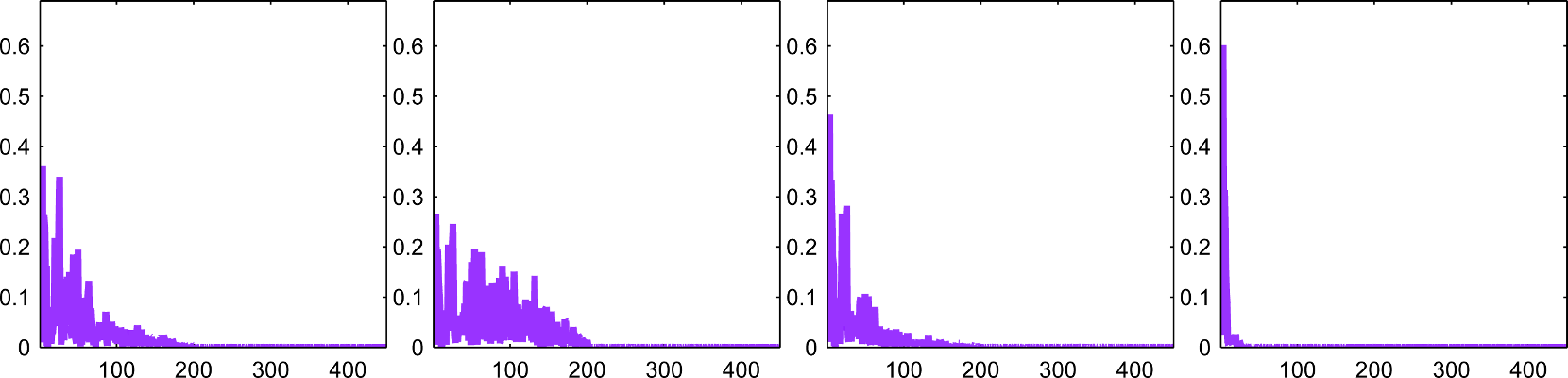}
	\caption{The Fourier coefficients of the eigenvectors displayed in Figure \ref{fig:optimaleigspace}. \newline 
	Top row: the absolute value of the Fourier coefficients $\hat{\psi}_1$ of the eigenvector ${\psi}_1$ for the graph $G_1$ and the filter pairs $(f_1,g_1)$, $(f_2,g_2)$, $(f_3,g_3)$, and $(f_4,g_4)$ (from left to right). \newline
	Bottom row: the absolute value of the Fourier coefficients $\hat{\phi}_1^{(\theta)}$ of the eigenvector $\phi_1^{(\theta)}$ ($\theta = \frac{9}{20} \pi$) for the graph $G_2$ and the filter pairs $(f_1,g_1)$, $(f_2,g_2)$, $(f_3,g_3)$, and $(f_4,g_4)$ (from left to right).}
	\label{fig:optimaleigfreq}
\end{figure}

Finally, we compare the space-frequency behavior of the optimally space-frequency localized eigenvectors $\psi_1$ and $\phi_1^{(\theta)}$ for the four filter pairs in Section \ref{sec:experimentfilters}. The spatial and spectral distributions of these localized eigenvectors are illustrated in Figure \ref{fig:optimaleigspace} and Figure \ref{fig:optimaleigfreq}, respectively.

Regarding the space localization, all four filter pairs provide eigenvectors $\psi_1$ and $\phi_1^{(\theta)}$ that are localized around the center node $w$ of the spatial filter. The kind of localization of the eigenvectors follows roughly the structure of the spatial filters given in Figure \ref{fig:2} and Figure \ref{fig:3}. In particular, whereas $f_1$ gives a set-oriented localization measure, the filters $f_2$, $f_3$ and $f_4$ are distance-oriented (with respect to the center $w$). The effects of the spectral filters on the eigenvectors get mainly visible in case of the filters $g_3$ and $g_4$. The decaying Fourier coefficients $\hat{g}_3$ and $\hat{g}_4$ have a blurring effect on the optimal eigenvectors, in particular in case of the pair $(f_4,g_4)$. 

In the spectral domain, we see that the bandlimiting filters $g_1 = g_2$ and $g_3$ are rather rough localization measures in the spectrum of the graph. In principle, they mainly push the optimal eigenvector to be in the given frequency band $\mathcal{B}$. On the other hand, the Laplace filter $g_4$ generates optimal eigenvectors with a much stronger frequency localization in the lower part of the spectral domain corresponding to the small eigenvalues of the graph Laplacian.

\section{Conclusion}
In this work, we presented a flexible framework for uncertainty relations in spectral graph theory that allows to characterize and compute uncertainty regions for a broad family of different space and frequency filters. In particular, the usage of a polygonal approximation method for the convex numerical range enabled us to  visualize the boundaries of the uncertainty regions very efficiently. This visualization technique and the related descriptions of uncertainty curves and space-frequency decompositions of signals make this framework into a promising tool to study and analyze new filter designs for a graph-adapted space-frequency analysis.

\section*{Acknowledgment}
The project was supported by the European Union's Horizon
2020 research and innovation programme ERA-PLANET, grant agreement no. 689443.

\end{document}